%% file: main.tex
\documentclass[12pt]{article}
\usepackage{amssymb,amsmath,amsthm,enumerate,comment}
\usepackage{calc}
\usepackage{graphicx, color}
    \usepackage[driverfallback=hypertex,pagebackref=true,colorlinks]{hyperref}
    \hypersetup{linkcolor=[rgb]{.7,0,0}}
    \hypersetup{citecolor=[rgb]{0,.7,0}}
    \hypersetup{urlcolor=[rgb]{.7,0,.7}}

\usepackage{mathrsfs}
\usepackage{tcolorbox}
\usepackage{geometry}
\usepackage{epstopdf}
\usepackage{cleveref}
\usepackage{natbib}

\usepackage{setspace}
\AtBeginDocument{%
  \addtolength\abovedisplayskip{-0.25\baselineskip}%
  \addtolength\belowdisplayskip{-0.15\baselineskip}%
}

\usepackage{float}
\floatstyle{ruled}

\usepackage[titletoc,title]{appendix}
\usepackage{url}

\usepackage{aliascnt} 

\newtheorem{theorem}{Theorem}[section]
\newtheorem*{theorem*}{Theorem}

\newaliascnt{definition}{theorem}
\newtheorem{definition}[definition]{Definition}
\aliascntresetthe{definition}

\newtheorem*{definition*}{Definition}

\newaliascnt{lemma}{theorem}
\newtheorem{lemma}[lemma]{Lemma}
\aliascntresetthe{lemma}

\newtheorem*{lemma*}{Lemma}

\newaliascnt{claim}{theorem}

\aliascntresetthe{claim}

\newtheorem*{claim*}{Claim}

\newaliascnt{fact}{theorem}

\aliascntresetthe{fact}

\newtheorem*{fact*}{Fact}

\newaliascnt{observation}{theorem}
\newtheorem{observation}[observation]{Observation}
\aliascntresetthe{observation}

\newtheorem*{observation*}{Observation}

\newaliascnt{conjecture}{theorem}

\aliascntresetthe{conjecture}

\newtheorem*{conjecture*}{Conjecture}

\newaliascnt{corollary}{theorem}

\aliascntresetthe{corollary}

\newtheorem*{corollary*}{Corollary}

\newaliascnt{remark}{theorem}

\aliascntresetthe{remark}

\newtheorem*{remark*}{Remark}

\newaliascnt{proposition}{theorem}
\newtheorem{proposition}[proposition]{Proposition}
\aliascntresetthe{proposition}

\newtheorem*{proposition*}{Proposition}

\def\R{{\mathbb {R}}}
\def\F{{\mathcal {F}}}
\def\C{{\mathcal {C}}}

\newcommand{\eps}{\varepsilon}
\newcommand{\hs}{\mathsf{HS}}

\def\sign{{\mathsf {sign}}}

\usepackage{algorithm,tikz,pgfplots}
\usetikzlibrary{shapes, arrows, patterns, calc, positioning}
\usetikzlibrary{shapes,arrows}
\usepackage{algpseudocode}
\usepackage{etoolbox}

\usepackage{bbm,bm}

\usepackage{xcolor}

\allowdisplaybreaks

\newcommand{\ip}[1]{\langle #1 \rangle}

\newcommand{\prob}[1]{\mathsf{#1}}
\renewcommand{\P}{\mathcal{P}}
\newcommand{\Q}{\mathcal{Q}}
\DeclareMathOperator{\conv}{conv}

\DeclareMathOperator{\dom}{dom}

\DeclareMathOperator{\D}{D}  

\makeatletter
\def\moverlay{\mathpalette\mov@rlay}
\def\mov@rlay#1#2{\leavevmode\vtop{%
   \baselineskip\z@skip \lineskiplimit-\maxdimen
   \ialign{\hfil$\m@th#1##$\hfil\cr#2\crcr}}}
\newcommand{\charfusion}[3][\mathord]{
    #1{\ifx#1\mathop\vphantom{#2}\fi
        \mathpalette\mov@rlay{#2\cr#3}
      }
    \ifx#1\mathop\expandafter\displaylimits\fi}
\makeatother

\title{Convex Set Disjointness, Distributed Learning of Halfspaces, and LP Feasibility}
\author{
Mark Braverman
\thanks{Princeton University}
\and
Gillat Kol
\thanks{Princeton University}
\and
Shay Moran
\thanks{Google AI Princeton}
\and
Raghuvansh R.\ Saxena
\thanks{Princeton University}}

\date{} 
\begin{document}

\maketitle
\thispagestyle{empty}

\begin{abstract}

We study the {\em Convex Set Disjointness} (CSD) problem, 
	where two players have input sets taken from an arbitrary fixed domain~$U\subseteq \mathbb{R}^d$ of size $\lvert U\rvert = n$. 
	Their mutual goal is to decide using minimum communication 
	whether the convex hulls of their sets intersect 
	(equivalently, whether their sets can be separated by a hyperplane). 

Different forms of this problem naturally arise in distributed learning and optimization:
	it is equivalent to {\em Distributed Linear Program (LP) Feasibility} -- a basic task in distributed optimization, and 
	it is tightly linked to {\it Distributed Learning of Halfdpaces in $\R^d$}. 
	In {communication complexity theory}, 
	CSD can be viewed as a geometric interpolation 
	between the classical problems of {Set Disjointness} (when~$d\geq n-1$) and {Greater-Than} (when $d=1$). 

We establish a nearly tight bound of $\tilde \Theta(d\log n)$ on the communication complexity of learning halfspaces in $\R^d$. 
	For Convex Set Disjointness (and the equivalent task of distributed LP feasibility)
	we derive upper and lower bounds of $\tilde O(d^2\log n)$ and~$\Omega(d\log n)$.
	These results improve upon several previous works in distributed learning and optimization.
	
Unlike typical works in communication complexity, the main technical contribution of this work lies in the upper bounds.
	In particular, our protocols are based on a {\it Container Lemma for Halfspaces} and on two variants of {\it Carath\'eodory's Theorem}, 
	which may be of independent interest. 
	These geometric statements are used by our protocols to provide a compressed summary of the players' input.

\end{abstract}

%
%

\input{intro.tex}

\input{epscoversbetter2.tex}

\input{prelim.tex}

\input{ub.tex}

\input{lb.tex}

\bibliographystyle{plainnat}
\bibliography{ref}

\newpage
\begin{appendices}
\input{app1.tex}  
 \end{appendices}

 \end{document}

%% file: intro.tex

\section{Introduction} \label{sec:intro}

Let $U\subseteq\R^d$ be an arbitrary set of $n>>d$ points and consider the {\it Convex Set Disjointness} communication problem 
$\prob{CSD}_U$ in which two parties, called Alice and Bob, hold input sets~$X,Y\subseteq U$
and their goal is to decide whether $\conv(X)\cap\conv(Y)=\emptyset$, where $\conv(\cdot)$ denotes the {convex hull} operator.
As we briefly discuss next, this problem has roots in {\it distributed learning}, {\it distributed optimization}, and in {\it communication complexity}.

\subsubsection*{Distributed Learning}
Some modern applications of machine learning involve collecting data from several sources.
For example, in healthcare related applications, data is often collected from hospitals and labs in remote locations.
Another host of examples involves algorithms that are trained on personal data  
(e.g.\ a music recommendation app which is trained on preferences made by numerous users).

Such applications raise the need for {\it algorithms that are able to train on distributed data 
without gathering it all on single a centralized machine}.
Moreover, distributed training is also beneficial from a privacy perspective
in contexts where the data contains sensitive information (e.g.\ personal data on smartphones).
Consequently, tech companies invest significant efforts in developing suitable technologies;
one notable example is Google's {\it Federated Learning} project~\citep{Konecny16federated}.

The Convex Set Disjointness communication problem was introduced in this context by \cite{kane17communication}  to analyze 
the communication complexity of learning linear classifiers.
Linear classifiers (a.k.a.\ halfspaces) form the backbone of many popular learning algorithms:
they date back to the seminal {\it Perceptron algorithm} from the 50's~\citep{Rosenblatt58perceptron},
and also play a key role in more modern algorithms such as kernel machines and neural nets.

In the distributed setting, Learning Halfspaces refers to the following task:
a set of {\it examples} is distributed between several parties. 
Each example consists of a pair~$(x,y)$, where~$x\in U$ is a feature vector, $y = \sign(L(x))$ is the label,
and $L:\R^d\to \R$ is the (unknown) target linear function.
The parties' goal is to agree on a classifier $h:U\to\{\pm 1\}$ such that~$h(x)=y$ for every input example~$(x,y)$,
while minimizing the amount of communication.
In this context, it may be natural to think of the domain $U$ as a grid, or as a discretized manifold, 
or any other domain that arises naturally from euclidean representations of data.

\paragraph{Our Contribution.}
{\it We provide a nearly tight bound of $\tilde \Theta(d\log n)$ 
	on the communication complexity of this problem  in the two-party setting.}  
	Our upper bound improves upon a previous bounds of $O(d\log ^2 n)$ by~\cite{Daume12efficient} and \cite{Balcan12dist} 
	which rely on distributed implementations of boosting algorithms. 
	Our protocol exploits a tool we call {\it halfspace containers}
	which may be of independent interest (\Cref{thm:onesidedcover} below).
	Roughly speaking, halfspace containers provide a way to summarize
	important information about the players' input in a compressed manner.

{\it We also give a nearly matching lower bound of $\Omega(d\log n)$}, 
	which improves upon a previous lower bound of  $\Omega(d+ \log n)$ by~\cite{kane17communication}.

Our upper bound is achieved by a deterministic protocol whereas our lower bound
	applies even when the protocol is randomized and may err with constant probability.

\subsubsection*{Distributed Optimization}

Linear Programming (LP) is one of the most basic primitives in optimization.
	In the associated decision problem, called LP feasibility, 
	the goal is to determine whether a system of linear inequalities (also called constraints) is satisfiable.
	In distributed LP feasibility the constraints are divided between several parties.

This problem is essentially equivalent to Convex Set Disjointness, 
	albeit in a dual formulation where constraints and points are interchanged: 
	indeed, disjointness of the convex hulls amounts to the existence of a separating hyperplane
	which, from a dual perspective, corresponds to point that satisfies all of the constraints.

\paragraph{Our Contribution.}	
{\it This work yields a protocol for LP feasibility in the two-party setting which communicates~$\tilde O(d^2 \log n)$ bits}.
	Similarly to our learning protocol, also this protocol is based on {\it halfspace containers} (\Cref{thm:onesidedcover}).
	This improves upon two incomparable previous upper bounds by \cite{vempala19optimization}
	which rely on classical sequential LP algorithms:
	(i)~a distributed implementation of \cite{Clarkson95lasvegas}'s algorithm 
	with communication complexity of $O(d^3\log^2 n)$ bits (see their Theorem 10.1),
	and (ii) a protocol based on the {\it Center of Gravity} algorithm (see their Theorem 11.3 ).
	The communication complexity of the latter protocol matches our $\tilde O(d^2\log n)$ 
	bound when the domain $U$ is a grid (e.g.\ $U = [n^{1/d}]^d$), but can\footnote{
	In fact, already in the one-dimensional case, if the domain $U\subseteq \R$ consists of~$n$ points which form a geometric 	
	progression (say $U=\{1,2,4, \ldots, 2^n\}$), then the Center of Gravity protocol can transmit up to $\Omega(n)$ bits, 
	which is exponentially larger than the$O(\log n)$ optimal deterministic protocol, and double exponentially larger than 
	the $O(\log\log n)$ optimal randomized protocol.} be significantly larger when $U$ is arbitrary.
	
{\it We also give a lower bound of $\Omega(d \log n)$}
	which is off by a factor of $d$ from our upper bound.
	Our lower bound applies also to randomized protocols that may err with a small probability.
	This improves upon \cite{vempala19optimization} 
	who derive a similar lower bound of~$\Omega(d\log n)$ in the deterministic setting (see their Theorem 3.6)
	and a lower bound of $\Omega({\log n})$ in the randomized setting (their Theorem 9.2).


\subsubsection*{Communication Complexity}
Convex Set Disjointness can be seen as a geometric interpolation between {\it Set Disjointness} (when $d \geq n-1$),
and {\it Greater-Than} (when $d=1$). 
Indeed, if~$d\geq n-1$ then one can pick the $n$ points in $U\subseteq\R^d$ to be affinely independent, which implies that
\[X\cap Y = \emptyset \iff \conv(X)\cap\conv(Y)=\emptyset.\]
Therefore, in this case the communication complexity of~$\prob{CSD}_U$ is the same like Set Disjointness which is $\Theta(n)$ \citep{kalyanasundaram1992probabilistic}. 
In the other extreme, if $d=1$ then $U$ is a set of $n$ points on the real line and~$\prob{CSD}_U$ 
boils down to comparing the two extreme points in Alice's input with the two extreme points in Bob's input (see \Cref{fig:1D}).
Thus, the case of $d=1$ is equivalent to the {Greater-Than} problem on $\log n$
bits, whose deterministic communication complexity is $\Theta(\log n)$ in the deterministic setting
and  $\Theta(\log\log n)$ in the randomized setting (with constant error) \citep{Feige94computing, Viola13addition}.

\begin{figure}

\begin{center}
\begin{tikzpicture}[vb/.style = {inner sep = 0.8mm, draw, fill, blue, circle}, vbnf/.style = {inner sep = 2mm, draw, blue, circle, ultra thick}, vg/.style = {inner sep = 0.8mm, draw, fill, red, circle}, vgnf/.style = {inner sep = 2mm, draw, red, circle, ultra thick}, scale = 0.8]

        \coordinate (c1) at (-8.5,0) {};
        \coordinate (c2) at (8.5,0) {};
        \draw[thin, black] (c1) -- (c2);
        
        \node[vg] (p1) at (8,0) {};
        \node[rectangle, below = 0.3cm of p1] (p11) {$\bm{y_{\text{\bf right}}}$};
        \node[vg] (p2) at (7,0) {};
        \node[vg] (p3) at (6,0) {};
        \node[vg] (p4) at (5,0) {};
        \node[vb] (p5) at (4,0) {};
        \node[rectangle, below = 0.3cm of p5] (p55) {$\bm{x_{\text{\bf right}}}$};
        \node[vg] (p6) at (2,0) {};
        \node[vg] (p7) at (1,0) {};
        \node[vg] (p9) at (0.5,0) {};
        \node[vg] (p10) at (0,0) {};
        \node[rectangle, below = 0.3cm of p10] (p1010) {$\bm{y_{\text{\bf left}}}$};
        \node[vb] (p0) at (-1,0) {};
        \node[vb] (p13) at (-2,0) {};
        \node[vb] (p14) at (-2.5,0) {};
        \node[vb] (p15) at (-3,0) {};
        \node[vb] (p17) at (-4,0) {};
        \node[vb] (p16) at (-6,0) {};
        \node[vb] (p19) at (-7,0) {};
        \node[vb] (p18) at (-8,0) {};
        \node[rectangle, below = 0.3cm of p18] (p1818) {$\bm{x_{\text{\bf left}}}$};

\end{tikzpicture}
\end{center}

\caption{
Convex Set Disjointness in 1D: the convex hull of Alice's input (blue points) is disjoint from the convex hull of Bob's input (red points)
if and only if ${\bf x_{right}} < {\bf y_{left}}$ or ${\bf y_{right}} < {\bf x_{left}}$. 
Thus, this case amounts to deciding (2 instances of) the {\it Greater-Than} problem on $\log n$ bits.
}\label{fig:1D}
\end{figure}
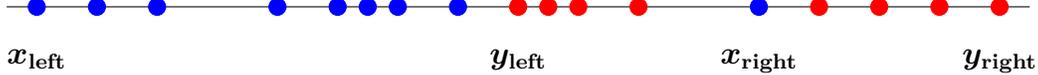


\begin{table}[]
	\centering
	\begin{tabular}{|l|l|l|}
		\hline \textbf{Dimension} & \textbf{Upper bound}           & \textbf{Lower bound}                 \\
		\hline $ d=1$; deterministic   & $O(\log n)$ [trivial]                    & $\Omega(\log n)$ [folklore]                         \\
		$d=1$; randomized  & $O(\log\log n)$ [\cite{Feige94computing}]      & $\Omega(\log\log n)$ [\cite{Viola13addition}]                 \\
		$d > 1$; deterministic  & $\tilde O(d^2\log n)$ {\bf [this work]}  & $\tilde \Omega(d \log n)$ [\cite{vempala19optimization}] \\
		$d > 1$; randomized & ~~~\texttt{"}~~~ & $\tilde \Omega(d \log n)$ {\bf [this work]} \\ \hline
	\end{tabular}
\caption{Deterministic and randomized communication complexity of $\prob{CSD}_U$ for arbitrary~$U\subseteq \R^d$ with $\lvert U\rvert =n$.
The case of $d=1$ is equivalent to the {\it Greater-Than} problem on $\log n$ bits.}
\label{tbl:bounds}
\end{table}

\subsubsection*{Organization}
We begin by formally stating the main results in \Cref{sec:results}.
Then, in \Cref{sec:related} we survey some of the related work.
\Cref{sec:overview} contains an overview of some of the proofs,
and \Cref{sec:proofs,sec:epscovers} contain the complete proofs.

\section{Results}\label{sec:results}
We begin with formally stating our results for Learning Halfspaces and for Convex Set Disjointness.
Later, in \Cref{sec:technical}, we present the halfspace container lemma along with  some geometric statements that arise 
in our analysis which may be of independent interest.

We use standard notation and terminology from communication complexity \citep{Kushilevitz97book}.
Specifically, for a boolean function $f$, let $\D(f)$ and $R(f)$ denote its deterministic and randomized\footnote{With error probability $\eps=1/3$.} communication complexity.

\subsection{Learning Halfspaces}

We first define the Halfspace Learning Problem.
Let $U\subseteq\R^d$ be a domain with $n$ points.
 An {\it example} is a pair of the form $(u,b)\in U\times\{\pm 1\}$. An example $(u,b)$ is called {\it positive}
if $b=+1$ and {\it negative} if $b=-1$. A set of examples $S\subseteq U\times\{\pm 1\}$ is called a sample.
{\it Learning Halfpaces over $U$} refers to the following search problem.
Alice's and Bob's inputs are samples~$S_a,S_b \subseteq U\times\{\pm 1\}$
such that there exists a hyperplane that separates the positive examples in $S_a\cup S_b$ 
from the negative examples in $S_a\cup S_b$.
Their goal is to output a function~$f:U\to \{\pm 1\}$ such that $f(x) = y$
for every example~$(x,y)\in S_a\cup S_b$.
If the protocol always outputs $f$ such that $f$ is an indicator of a halfspace
then the protocol is called a {\it proper} learning protocol. 
Otherwise it is called an {\it improper} learning protocol.

The following theorems establish a bound of $\tilde\Theta(d\log n)$ 
on the communication complexity of Learning Halfspaces.

\begin{theorem}[Upper bound]\label{thm:ublearning}
Let $d,n\in\mathbb{N}$, and let $U\subseteq \R^d$ be a domain with $n$ points. 
Then, there exists a deterministic protocol for Learning Halfspaces over $U$ with communication complexity~$O(d\log d\log n)$.
\end{theorem}

We note that our protocol is improper. 
It remains open whether the above bound can be achieved by a proper protocol.

\begin{theorem}[Lower bound]\label{thm:lblearning} Let $d,n \in\mathbb{N} $. 
Then, there exists a domain $U\subseteq \R^d$ with~$n$ points such that every 
(possibly improper and randomized) protocol that learns halfspaces over $U$ must transmit at least $\Omega(d\log(n/d))$ bits 
of communictaiton.
\end{theorem} 

\Cref{thm:ublearning,thm:lblearning} are proved in \Cref{sec:proofs}.
A proof overview is given in \Cref{sec:overviewlearning}

\subsection{Convex Set Disjointness and LP Feasibility}
Recall that $\prob{CSD}_U$ denotes the Convex Set Disjointness problem on a domain $U\subseteq\R^d$.

\begin{theorem}[Upper bound]\label{thm:mainub}
Let $d,n\in\mathbb{N}$, and let $U\subseteq \R^d$ be a domain with $n$ points. Then,
\[\D(\prob{CSD}_U) = O(d^2\log d \log n).\]
\end{theorem}

\begin{theorem}[Lower bound]\label{thm:mainlb}
Let $d,n\in\mathbb{N}$. Then, there exists a domain $U\subseteq \R^d$ with $n$ points such that
\[R(\prob{CSD}_U) = \Omega\bigl(d \log (n/d)\bigr).\]
\end{theorem}

As noted in the introduction, Convex Set Disjointness is equivalent to distributed LP feasibility,
and therefore the above bounds apply in both contexts.
\Cref{thm:mainub,thm:mainlb} are proved in \Cref{sec:proofs}.
A short overview of the proofs is given in \Cref{sec:overviewdisj}

\subsection{Geometric Results}\label{sec:technical}

Our analysis utilizes some geometric tools which, to the best of our knowledge, are novel.
As some of them may be of independent interest, we next present them in a self contained manner.

\subsubsection{Halfspace Containers}\label{sec:containers}
Our protocols hinge on {\it $\eps$-containers\footnote{This notation is inspired by a similar notion that arises in Graph Theory (see, {\em e.g.}, \cite{Balogh18containers} and references within).}} (defined below).
This is a variant of the notion of {\it $\eps$-covers}, which we recall next:
an {\it $\eps$-cover} for a family $\F\subseteq 2^X$ is a family~$\C \subseteq 2^X$ 
such that for every $F\in \F$ there is $C\in\C$ such that the symmetric difference\footnote{Equivalently, the hamming distance between the indicator vectors.} between~$C$ and $F$ 
is of size at most $\eps \lvert X\rvert$.
In other words, the hamming balls of radius $\eps\lvert X\rvert$ around $\C$ cover $\F$.
Note that this is a special instance of the notion of $\eps$-cover in metric spaces.
In the case of containers, we also require that $F\subseteq C$:
\begin{definition}[Containers]\label{def:epscover}
Let $X$ be a finite set and let $\F\subseteq 2^{X}$ be a family of subsets.
A family $\C\subseteq 2^{X}$ is a family of $\eps$-containers for $\F$ if
\[(\forall F\in \F) (\exists C\in \C)~: ~F\subseteq C~ \text{ and }~ \lvert C\setminus F\rvert \leq \eps \lvert X\rvert.\]
\end{definition}
Note that every set of $\eps$-containers is in particular an $\eps$-cover (but not vice versa).

\paragraph{A Container Lemma for Halfspaces.}
Let $\hs_d$ denote the family of all halfspaces in $\R^d$,
and for~$U\subseteq \R^d$ let $\hs(U) = \{H\cap U : H\in\hs_d\}$ denote the family of all halfspaces restricted to~$U$.
A classical result by Haussler implies that $\hs(U)$
has an $\eps$-cover of size roughly~$(1/\eps)^d$~\citep{haussler1995sphere}.
A remarkable property of this $\eps$-cover is that its size
depends only on $\eps$ and $d$;  in particular, it does not depend on~$\lvert U\rvert$.

The following result, which is our main technical contribution, 
establishes a similar statement for $\eps$-containers.

\begin{theorem}[Container Lemma for Halfspaces]\label{thm:onesidedcover}
Let $U\subseteq \R^d$. Then, for every $\eps > 0$ there is a set of $\eps$-containers for $\hs(U)$ of size $(d/\eps)^{O(d)}$.
\end{theorem}
We mention that, in contrast with Haussler's result which applies to any family with VC dimension $d$,
Theorem~\ref{thm:onesidedcover} does not extend to arbitrary VC classes (e.g.\ it fails for projective planes; see \Cref{sec:epscovers}).
This is also reflected in our proof which exploits geometric properties of halfspaces,
and in particular a dual version of Carath\'eodory's Theorem (see~\Cref{lem:dualcarath} below).
We discuss it in more detail in Section~\ref{sec:epscovers}, where we also prove Theorem~\ref{thm:onesidedcover}.

\subsubsection{Variants of Carath\'eodory's Theorem}
Carath\'eodory's Theorem is a fundamental statement in convex geometry \citep{Caratheodory07}:
	it asserts that if $x\in \R^d$, $Y\subseteq \R^d$ are such that $x\in\conv(Y)$ 
	then there are~$y_{1},\ldots,y_{d+1}\in Y$ such that $x\in\conv(\{y_{1},\ldots,y_{d+1}\})$.
	Our proof of \Cref{thm:mainub} exploits two variants of Carath\'eodory's Theorem.

\paragraph{A Dual Variant.}
Let $\Q\subseteq\R^d$ be a polytope. 
There are two natural ways of representing~$\Q$:
(i)~as the convex hull of its vertices, (ii) as an intersection of halfspaces.

Carath\'eodory's Theorem implies that if $\Q$
is the convex hull of a few vertices then it can be covered by a few simplices:
indeed,  if $\Q$ has $n$ vertices then, by Carath\'eodory's Theorem, 
 it can be covered by at most~$n^{d+1}$ 
sets of the form $\conv(\{x_{1},\ldots,x_{d+1}\})$, where the $x_i$'s
are vertices of~$\Q$.

Assume now that $\Q$ is an intersection of few halfspaces (say $n$).
How many subsimplices are needed in order to cover $\Q$ in this case?
A bound of $n^{d(d+1)}$ follows by the previous bound, 
since the number of vertices in $\Q$ is at most $n^d$ (every vertex is defined by $d$ hyperplanes). 
The next proposition achieves a quadratic improvement in the exponent.

\begin{proposition}[A dual variant of Carath\'eodory's Theorem]\label{lem:dualcarath}
Let $\Q\subseteq\R^d$ be a polytope that can be represented as an intersection of $n$ halfspaces.
Then, $\Q$ can be covered using at most~$n^{d}$ subsimplices of the form $\conv(\{x_{0},\ldots,x_{d}\})$,
where the $x_i$'s are vertices of $\Q$.
\end{proposition}

\Cref{lem:dualcarath} is proven in \Cref{sec:epscovers}.

\paragraph{A Symmetric Variant. }
Carath\'eodory's Theorem concerns a relation between a point~$x$ and a set $Y$ such that $x\in\conv(Y)$. 
The following simple generalization provides a symmetric relation between two set $X,Y$ such that $\conv(X)\cap\conv(Y)\neq\emptyset$.
\begin{proposition}[A symmetric variant of Carath\'eodory's Theorem]\label{lem:bicaratheodory}
Let $X, Y\subseteq \R^d$ such that $\conv(X)\cap \conv(Y)\neq\emptyset$.
Then $\conv(S_1)\cap \conv(S_2)\neq\emptyset$ for some $S_1\subseteq X, S_2\subseteq Y$
such that $\lvert S_1\rvert + \lvert S_2\rvert \leq d+2$.
\end{proposition}
Note that Carath\'eodory's Theorem boils down to the case where $X=\{x\}$ (and hence $\conv(X)\cap \conv(Y)\neq\emptyset\implies x\in\conv(Y)$).

Since the proof of \Cref{lem:bicaratheodory} is short, we present it here.
\begin{proof}[Proof of \Cref{lem:bicaratheodory}]
The proof follows an argument similar to the linear algebraic proof of Carath\'eodory's Theorem.
Assume $z\in \conv(X)\cap \conv(Y)$ 
can be represented as a convex combination of $d_1$ points $x_1 \ldots x_{d_1}\in X$
and as a convex combination of $d_2$ points~$y_1 \ldots y_{d_2}\in Y$ such that $d_1+d_2 > d+2$.
Consider the system of linear equalities in~$d_1+ d_2$ variables $\alpha_1\ldots \alpha_{d_1},\beta_1\ldots \beta_{d_2}$
defined by the constraints (i) $\sum \alpha_i x_i = \sum \beta_j y_j $, and (ii) $\sum \alpha_i=\sum\beta_j=0$. 
This system has $d_1+d_2 > d+2$ variables and only $d+2$ constraints ($d$ constraints from (i) and 2 more constraints from (ii)). 
Thus, it has a solution such that not all $\alpha_i$'s and $\beta_j$'s are 0. 
Consequently, one can shift $z$ by a sufficiently small scaling of the vector~$v=\sum \alpha_i x_i = \sum \beta_i y_i$,
so that one of the coefficients of the $x_i$'s or the $y_j$'s vanishes.
This process can be repeated until~$d_1 + d_2 \leq d+2$, 
which yields the desired sets~$S_1\subseteq X, S_2\subseteq Y$.

\end{proof}

\noindent {\bf Remark.}
\Cref{lem:bicaratheodory} establishes a tight bound of $d+2$ on the {\it coVC number} of halfspaces in $\R^d$. 
The coVC number is a combinatorial parameter which characterizes the concept classes
that can be properly learned using polylogarithmic communication complexity
(see \cite{kane17communication}).
It is defined as follows: let $H\subseteq\{\pm 1\}^X$ be an hypothesis class over a domain $X$.
Its coVC number is the smallest number $k$ such that every sample $S\subseteq X\times\{\pm 1\}$
which is not realizable\footnote{A sample $S$ is realizable with respect to $H$ if there is $h\in H$ such that $h(x)=y$ for every $(x,y)\in S$.} by $H$ 
has a subsample $S'\subseteq S$ of size $\lvert S'\rvert \leq k$ which is not realizable by $H$.
A weaker upper bound of $2d+2$ on the coVC number of halfspaces was given by \cite{kane17communication} 
(see Example 1 in their paper).

\section{Related Work}\label{sec:related}

\cite{lovasz93communication} studied a variant of convex set disjointness
where the goal is to decide whether the convex hulls
intersect in a point from $U$. 
This variant exhibits a very different behaviour, even in dimension $d=2$:
indeed, if $U$ is in convex position\footnote{A set $U$ is in convex position if $u\notin\conv(U\setminus\{u\})$ for all $u\in U$. } (say $n$ points on the unit circle) then this becomes equivalent to the classical set disjointness problem whose communication complexity is $\Theta(n)$,
whereas in the formulation considered in this paper, any planar instance $U\subseteq \R^2$ can be decided using $O(\log n)$ bits.

Variants of the convex set disjointness problem were considered by several works in
{distributed machine learning} and {distributed optimization}
(see, {\em e.g.},~\cite{Balcan12dist,Daume12efficient,Chen16boosting,kane17communication,vempala19optimization}.
Other variants in which the number of rounds is bounded arise in space lower bounds for learning linear classifiers in streaming models~\citep{Dagan19space}.
 
\cite{kane17communication} studied convex set disjointess in a more general communication model 
in which the input domain $U$ may be infinite, and the players are allowed to transmit points from their input sets
for a unit cost of communication. They established an upper bound of $\tilde O(d^3\log n)$
and a lower bound of $\tilde\Omega(d + \log n)$ on the number of transmitted points/bits
when the input subsets are of size $n$ and the dimension is $d$. These bounds translate\footnote{The extra $\log n$ factor in the upper bound is because transmitting $u\in U$ requires $\log\lvert U\rvert = \log n$ bits. } to upper and lower bounds
of $\tilde O(d^3\log^2 n)$ and $\tilde\Omega(d + \log n)$ in the setting considered in this paper.

Recently, \cite{vempala19optimization} published a thorough study of
communication complexity of various optimization problems.
One of the problems they consider is Linear Program feasibility, 
which, as explained in the introduction, is equivalent to Convex Set Disjointness.
The main difference is that \cite{vempala19optimization} do not consider arbitrary domains $U$,
and focus on the case when $U$ is a grid (say $[n^{1/d}]^d$).
On the other hand, in our setting~$U$ can be arbitrary.
They derive a lower bound of $\Omega({\log n})$ in the randomized setting (Theorem 9.2)
and of $\Omega(d\log n)$ in the deterministic setting (Theorem 3.6), as well as several upper bounds. 
Their best upper bound of $O(d^2\log^2 d  \log n)$ (Theorem 11.3) is based on an implementation of the {\it Center of Gravity} algorithm.
This matches (up to an extra ``$\log d$'' factor) the upper bound given in this work.
However, their upper bound  does not apply\footnote{Specifically, their analysis exploits the assumed grid structure of $U$: 
their bound on the number of iterations of the protocol uses bounds on determinants of matrices with entries from $[n^{1/d}]$.} 
to arbitrary domains~$U$. 
In fact, already in the one-dimensional case, if the domain $U\subseteq \R$ consists of~$n$ points which form a geometric progression (say $U=\{1,2,4, \ldots, 2^n\}$),
then the Center of Gravity protocol can transmit up to $\Omega(n)$ bits, which is exponentially larger than the
$O(\log n)$ optimal deterministic protocol, and double exponentially larger than the $O(\log\log n)$ optimal randomized protocol.
It is worth noting that \cite{vempala19optimization} provide another upper bound (Theorem 10.1),
which is based on \cite{Clarkson95lasvegas}'s algorithm 
whose analysis extends arbitrary domains $U$.
This protocol has communication complexity of $O(d^3\log^2 n)$ bits 
(matching the bound of \cite{kane17communication}).

\section{Proofs Overview}\label{sec:overview}

In this section we overview the proofs and highlight some of the more technical arguments.
We begin with overviewing the proof of the Halfspace Container Lemma (\Cref{thm:onesidedcover}),
which is the most involved derivation in this work and forms the crux of our communication protocols. 
Then, we outline the proofs for Convex Set Disjointness in \Cref{sec:overviewdisj}
and for Distributed Halfspace Learning in \Cref{sec:overviewlearning}.

\subsection{Halfspace Containers}\label{sec:containeroverview}

Let $U\subseteq \R^d$ be a domain with $n$ points.
We want to show that for every $\eps>0$ there is a collection of (roughly) $(d/\eps)^d$ sets
called {\it containers} such that for every halfspace $H$ there is a container $C$ such that $H\subseteq C$
and $C\setminus H$ contains at most $\eps\cdot n$ points from $U$.
It will be more convenient to prove the following equivalent statement in which $H$ and $C$ switch roles:
\begin{center}
{\it There is a collection $\mathcal{C}$ of (roughly) $(d/\eps)^d$ sets
such that for every halfspace $H$ there is~$C\in \C$ such that \underline{$C\subseteq H$}
and $H\setminus C$ contains at most $\eps\cdot n$ points from $U$.}
\end{center}
Indeed, these statements are equivalent, because a complement of a halfspace is a halfspace,
and so taking the complements of all sets in a family $\C$ with the above property yields the desired family of $\eps$-containers.

\paragraph{Constructing an $\eps$-net.}
Each of the sets in the constructed family $\mathcal{C}$ will be an intersection of $d+1$ halfspaces.
The first step in the construction is to pick a ``small''  $V\subseteq U$ which forms an $\eps$-net
to sets of the form $H_0\setminus (\cap_{i\leq d+1} H_i)$, where the $H_i$'s are halfspaces:
\begin{center}
{\it That is, $V$ satisfies that for every set $B$ of the form $B=H_0\setminus (\cap_{i\leq d+1} H_i)$, 
if $B$ contains at least $\eps\cdot n$ points from $U$
then~$B\cap V \neq\emptyset$.}
\end{center}
By standard arguments from VC theory, a random subset $V\subseteq U$ of size roughly $d^2/\eps$
will satisfy this property.
Once we have such an $\eps$-net $V$, the idea is to associate with any given half-space $H$
a set of $d+1$ halfspaces $H_1,\ldots H_{d+1}$ which are induced by $V$ such that
\begin{itemize}
\item[(i)] $\cap_{i\leq d+1} H_i \subseteq H$, and 
\item[(ii)] $H\setminus (\cap_{i\leq d+1} H_i)$ does not contain any point from $V$.
\end{itemize}
Since $V$ is an $\eps$-net, property (ii) implies that $H\setminus (\cap_{i\leq d+1} H_i)$
contains at most $\eps\cdot n$ points from $U$, as needed.

\paragraph{Dual Polytope.}

To derive the halfspaces $H_1,\ldots H_{d+1}$ which satisfy the above properties (i) and (ii) 
we consider the {\it dual space} in which each halfspace is associated by a $d+1$ dimensional vector
of the form $(\vec a,b)$, where $\vec a\in \R^d$ is the normal to the supporting hyperplane and $b$ is the bias;
that is, the halfspace is given by $\{x\in\R^d : \vec a\cdot x \leq b\}$.

Consider the set $\P=\P(H)$ of all halfspaces that are {\it equivalent} to $H$
with respect to the $\eps$-net $V$. That is, $\P\subseteq \R^{d+1}$ contains representations
of all halfspaces $H'$ such that~$H'\cap V= H\cap V$ (we stress that there can be several such halfspaces which
have a different intersection with the domain $U$). 
Note that $\P$ is a convex set which is defined\footnote{In the complete proof 
we will define $\P$ with $O(d)$ more constraints in order to ensure boundedness.} by $\lvert V\rvert$ linear inequalities
(each $v\in V$ corresponds to a linear inequality posing that $v\in H\iff v\in H'$).
For an illustration, see \Cref{fig:auxpoly}.

\begin{figure}
\begin{center}
\begin{tikzpicture}[v/.style = {inner sep = 0.5mm, draw, fill, black, circle}, a/.style = {inner sep = 0.5mm, draw, fill, blue, circle}, e/.style = {blue, draw, dashed, ultra thick}, f/.style = {thick}, f1/.style = {ultra thick}, vb/.style = {inner sep = 0.8mm, draw, fill, blue, circle}, vbnf/.style = {inner sep = 2mm, draw, blue, circle, ultra thick}, vg/.style = {inner sep = 0.8mm, draw, fill, red, circle}, vgnf/.style = {inner sep = 2mm, draw, red, circle, ultra thick}, scale = 0.8]
	\begin{scope}[shift = {(0, -14, 0)}, scale = 0.6]
        \coordinate (v1) at (0,0,0) {};
        \coordinate (v2) at (9,1,-2) {};
        \coordinate (v3) at (2,9,-5) {};
        \coordinate (v4) at (8,1,-9) {};
        \coordinate (v5) at (4,9,-9) {};
        \coordinate (v8) at (9,9,-9) {};

        \coordinate (label) at (6,13,-5) {};
        \end{scope}

        \coordinate (a3) at (v8) {};
        \coordinate (a2) at ($(v3)!0.5!(a3)$) {};
        \coordinate (a1) at ($(v2)!0.8!(a2)$) {};
        \coordinate (a0) at ($(v1)!0.7!(a1)$) {};
        

        \node[v, label = -90:{\large $ v_1$}] (w1) at (v1) {};
        \node[v, label = -90:{\large $ v_2$}] (w2) at (v2) {};
        \node[v, label = 90:{\large $ v_3$}] (w3) at (v3) {};
        \node[v, label = 90:{\large $ v_4$}] (w4) at (v8) {};
        \node[a, label = 180:{\large $H$}] (c0) at (a0) {};
        
        \draw[f1] (v1) -- (v2) -- (v3) -- cycle;
        \draw[f1] (v1) -- (v3) -- (v5) -- cycle;
        \draw[f1] (v1) -- (v5) -- (v4) -- cycle;
        \draw[f1] (v1) -- (v4) -- (v2) -- cycle;
        \draw[f1] (v8) -- (v2) -- (v3) -- cycle;
        \draw[f1] (v8) -- (v3) -- (v5) -- cycle;
        \draw[f1] (v8) -- (v5) -- (v4) -- cycle;
        \draw[f1] (v8) -- (v4) -- (v2) -- cycle;

        \draw[f, blue, pattern = north west lines, pattern color = blue!30] (v1) -- (v2) -- (v3) -- cycle;
        \draw[f, blue, pattern = north west lines, pattern color = blue!30] (v1) -- (v3) -- (v8) -- cycle;
        \draw[f, blue, pattern = north west lines, pattern color = blue!30] (v1) -- (v8) -- (v2) -- cycle;
        \draw[f, blue, pattern = north west lines, pattern color = blue!30] (v3) -- (v2) -- (v8) -- cycle;

	\node[circle] (label2) at (label) {\Huge $\bm{P}$};

	\node (face1) at (barycentric cs:v1=1,v2=1,v3=1,v4=0,v5=0,v8=0) {};
	\node (face2) at (barycentric cs:v1=1,v2=0,v3=1,v4=0,v5=1,v8=0) {};
	\node (face3) at (barycentric cs:v1=1,v2=0,v3=0,v4=1,v5=1,v8=0) {};
	\node (face4) at (barycentric cs:v1=1,v2=1,v3=0,v4=1,v5=0,v8=0) {};

        \node[vb] (p4) at (3,3.4) {};
        \node[vb] (p5) at (2.5,3.6) {};
        \node[vb] (p7) at (1.5,2.5) {};
        \node[vb] (p8) at (1,1.8) {};
        \node[vb] (p9) at (0.5,2.9) {};
        \node[vb] (p0) at (-0.5,1.6) {};
        \node[vb] (p12) at (-1.5,3.4) {};
        \node[vb] (p13) at (-2,1.1) {};
        \node[vb] (p14) at (-2.5,0.3) {};
        \node[vb] (p15) at (-3,2.4) {};
        \node[vb] (p16) at (-3.5,1.2) {};
        \node[vb] (p17) at (-4,0.3) {};
        \node[vbnf, label = 90: {\large $a_1$}] (pp11) at (1,1.8) {};
        \node (aa11) at ($ (pp11) + (0.3,-0.3) $) {};
        \node[vbnf, label = 90: {\large $a_2$}] (pp22) at (-2.5,0.3) {};
        \node (aa22) at ($ (pp22) + (0.3,-0.3) $) {};
        
        \node[vg] (n6) at (4,-0.5) {};
        \node[vg] (n7) at (3.5,0.1) {};
        \node[vg] (n8) at (3,-1.7) {};
        \node[vg] (n10) at (2,-1.4) {};
        \node[vg] (n11) at (1.5,-1.1) {};
        \node[vg] (n12) at (1,-1.9) {};   
        \node[vg] (n13) at (0.5,-3.0) {};
        \node[vg] (n14) at (-0.5,-2.4) {};
        \node[vg] (n15) at (-1,-2.7) {};
        \node[vg] (n16) at (-2.2,-2.3) {};
        \node[vg] (n17) at (-2,-5.1) {};
        \node[vg] (n18) at (-2.5,-4.2) {};
        \node[vgnf, label = 270: {\large $b_1$}] (nn11) at (3.5,0.1) {};
        \node (bb11) at ($ (nn11) + (-0.3,0.3) $) {};
        \node[vgnf, label = 270: {\large $b_2$}] (nn22) at (-2.2, -2.3) {};
        \node (bb22) at ($ (nn22) + (-0.3,0.3) $) {};

        \coordinate (actual1) at ($(aa11)!0.5!(bb11)$);
        \coordinate (actual2) at ($(aa22)!0.5!(bb22)$);
        \draw[ultra thick, black, shorten >=-6cm,shorten <=-6cm] (actual1) -- node [below, pos = 2.5] {{\large $\bm{H}$}} (actual2);
        \coordinate (actual) at ($(actual1)!2.4!(actual2)$) {};
        \draw[ultra thick, ->] (actual) -- node [left] {{\large $\bm{+}$}} ($(actual)!1cm!90:(actual2)$);
        
        \draw[f, black, dashed, shorten >=-6cm,shorten <=-6cm] (aa11) -- (aa22);
        \coordinate (aa12n) at ($(aa11)!1.8!(aa22)$) {};
        \draw[dashed, ->] (aa12n) -- node [left] {$\bm{+}$} ($(aa12n)!1cm!90:(aa22)$);
        
        \draw[f, black, dashed, shorten >=-6cm,shorten <=-6cm]  (aa22) -- (bb11);
        \coordinate (ab21n) at ($(aa22)!-1!(bb11)$) {};
        \draw[dashed, ->] (ab21n) -- node [left] {$\bm{+}$} ($(ab21n)!1cm!90:(bb11)$);
        
        \draw[f, black, dashed, shorten >=-6cm,shorten <=-6cm]  (bb11) -- (bb22);
        \coordinate (bb12n) at ($(bb11)!1.8!(bb22)$) {};
        \draw[dashed, ->] (bb12n) -- node [left] {$\bm{+}$} ($(bb12n)!1cm!90:(bb22)$);
        
        \draw[f, black, dashed, shorten >=-6cm,shorten <=-6cm]  (bb22) -- (aa11);
        \coordinate (ba21n) at ($(bb22)!-1!(aa11)$) {};
        \draw[dashed, ->] (ba21n) -- node [left] {$\bm{+}$} ($(ba21n)!1cm!90:(aa11)$);

        
        \node[circle, below = 2cm of face1] (facelabel1) {$\bm{b_1}$};
        \node[circle, left = 2cm of face2] (facelabel2) {$\bm{b_2}$};
        \node[circle, right = 4cm of face3] (facelabel3) {$\bm{a_1}$};
        \node[circle, below = 2cm of face4] (facelabel4) {$\bm{a_2}$};
        \path
        (facelabel1) edge [bend right, ->, thick] (face1)
        (facelabel2) edge [bend left, ->, thick] (face2)
        (facelabel3) edge [bend right, ->, thick] (face3)
        (facelabel4) edge [bend right, ->, thick] (face4);
        
\end{tikzpicture}
\end{center}
\caption{The auxiliary dual polytope: halfspaces in the top-left of the figure (the primal space) are represented by points in the bottom right of the figure (the dual space), and points in the top-left correspond to half-spaces in the bottom right. The circled points in the top-left denote the points in the $\eps$-net $V$; these points define the facets of the auxiliary polytope $\P$, which is (a dual representation of) the set of halfspaces that induce the same partition on $V$ like $H$.}
\label{fig:auxpoly}
\end{figure}
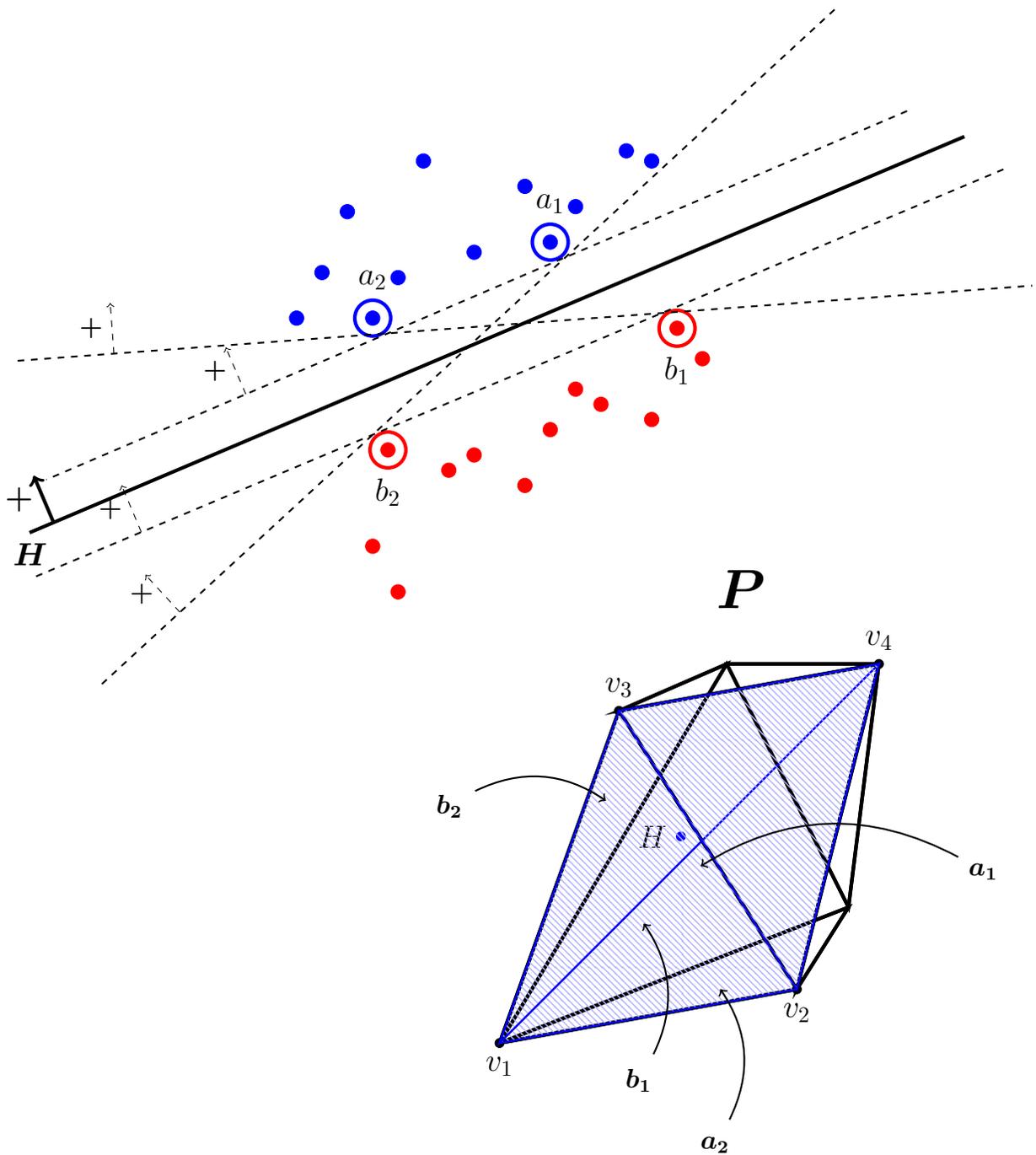

Now, by Carath\'eodory Theorem there are $d+2$ vertices of $\P$ such that $H$ is in their convex hull.
By the definition of $\P$, these $d+2$ vertices correspond to halfspaces $H_i\subseteq \R^d$ such that $H_i\cap V= H\cap V$.
We claim that these $H_i$'s satisfy the above properties (i) and (ii).
Indeed, since $H$ is in their convex hull it follows that $\cap_i H_i \subseteq H$ which amounts to (i), 
and since the $H_i$'s are in $\P$, we have that $H_i\cap V= H_i$ for every $i$ which implies (ii).

\paragraph{An Inferior Bound.}
Let us now see how to get an inferior bound of $\lvert V\rvert ^{O(d^2)} = (d/\eps)^{O(d^2)}$ on the size of $\C$.
How many polytopes $\P(H)$ are there? (counting over all possible halfspaces~$H$.)
The constraints defining each polytope $\P$ are determined by the intersection~$V\cap H$,
where $H$ is a halfspace. Therefore, since there are~$O(\lvert V\rvert^{d})$ distinct 
intersections of $V$ with halfspaces, we get that there are $O(\lvert V\rvert^d)$ such polytopes $\P(H)$.
Now, given a fixed $\P(H)$, how many vertices does it have?
$\P$ is defined by $\lvert V\rvert$ constraints and therefore has at most $\lvert V\rvert^{d+1}$ vertices 
(each vertex is determined by $d+1$ constraints).
Therefore the number of $d+2$ tuples of vertices is at most $\lvert V\rvert^{(d+1)(d+2)}$.
To conclude, the number of possibilities for obtaining the halfspaces $H_1\ldots H_{d+2}$ is bounded by
\[\lvert V\rvert^{O(d)}\cdot \lvert V\rvert^{O(d^2)} =(d/\eps)^{O(d^2)}. \]
To remove the extra factor of $d$ from the exponent we exploit the Dual Carath\'eodory Theorem (\Cref{lem:dualcarath}),
which enables us to find a collection of just $\lvert V\rvert^{O(d)}$ tuples of~$(d+2)$ vertices
such that every point in $P(H)$ is in the convex hull of one of these tuples.

\paragraph{Dual Carath\'eodory Theorem.}

The $\lvert V\rvert^{O(d^2)}$ dependence in the above calculation arises because
for every $H$, we count a tuple of $d+2$ vertices of $\P(H)$ whose convex hull contains~$H$.
As the number of vertices can be as large as roughly $\lvert V\rvert^d$, 
a naive counting such as the one sketched above yields  a bound of $\lvert V\rvert^{O(d^2)}$.
In order to improve this, it suffices to show that~$\P(H)$ can be covered
by $\lvert V\rvert^{O(d)}$ subsimplices (i.e.\ sets of the form $\conv(H_1,\ldots,H_{d+2})$ where $H_1\ldots H_{d+2}$ are vertices of~$\P(H)$).

To this end we prove \Cref{lem:dualcarath} which asserts more generally, 
that if a polytope~$\Q\subseteq \R^d$ is defined by $n$ linear inequalities then it can be covered by $n^{d}$ subsimplices
(in our context the number of constraints $n$ is $\lvert V\rvert \approx d^2/\eps$ and the dimension $d$ is $d+1$).
We prove this in a constructive manner using a process from computational geometry called {Bottom Vertex Triangulation}~\citep{Clarkson88queries,Goodman04handbook}. 

{\it In a nutshell, given a point $a\in \Q$ we use the Bottom Vertex Triangulation process to encode in a sequence of $d$ out of the $n$
linear inequalities that define $\Q$, the names of $d+1$ vertices of $\Q$ whose convex hull contains $a$.}
This implies that the polytope can be covered using at most $n^d$ subsimplices, corresponding
to the number of sequences of length $d$ out of a set of size $n$.

In more detail, the sequence is defined as follows (see~\Cref{fig:bvt} for an illustration).
Given the input point $a$, let $\bm{p}_0$ be the bottom-most\footnote{Or any other canonical vertex.} vertex of $\Q$, 
and shoot a ray starting in $\bm{p}_0$ which passes through $a$ until it hits a facet $\Q_1$ of $\Q$ in a point $a_1\in \Q_1$.
Append to the constructed sequence the name of the linear inequality which became tight as a result of hitting $\Q_1$. 
Next, continue recursively the same process on $\Q_1$
(i.e.\ again shoot a ray from its bottom vertex $\bm{p}_1$ which passes through~$a_1$ until it hits a facet $\Q_2$, etcetera).
We refer the reader to \Cref{fig:encode} to an illustration of this encoding procedure as well as to \Cref{fig:bvt} 
for an illustration of the bottom vertex triangulation process.

\begin{figure}
\begin{tcolorbox}
\begin{center}
{\bf A Dual Carath\'eodory's Theorem}\\
\end{center}
\underline{\bf Encoding:}

\noindent
{\em Input:} a polytope $\Q \in \R^d$ which is defined by $n$ constraints (linear inequalities) and a point $\bm{a}\in \Q$.\\
{\em Output:} a sequence $\bm{S}$ of $d$ constraints which encodes vertices $\bm{x}_0,\ldots,\bm{x}_{d}\in \Q$ such that $\bm{a}\in\conv(\{\bm{x}_0,\ldots,\bm{p}_{d}\})$.
\begin{enumerate}
\item[(1)] Initialize $\Q_0 = \Q$, $\bm{a}_0=\bm{a}$, $\bm{x}_0 = \bm{p}(\Q_0)$, and $\bm{S}=\bm{\eps}$ (the empty sequence).\\
($\bm{p}(\Q')$ denotes the bottom vertex of a polytope $\Q'$.)
\item[(2)]  For $i=1,\ldots,d$:
\begin{enumerate}
\item[(2.1)] Extend the ray that starts at $\bm{x}_{i-1}$ and passes through $\bm{a}_{i-1}$ until it hits the boundary of $\Q_{i-1}$.
\item[(2.2)] Set $\bm{a}_{i}$ to be the point on the boundary of $\Q_{i-1}$ that the ray hits. 
Set $\Q_{i}$ to be the\footnote{If $a_{i+1}$ belongs to several facets (i.e.\ it sits on a face whose dimension is $<d-1$) then pick $\Q_{i+1}$ to be any facet that contains it.} facet of $\Q_{i-1}$ that contains $\bm{a}_{i}$ and Set $\bm{x}_{i}=\bm{p}(\Q_{i+1})$.
\item[(2.3)] Append to $\bm{S}$ the linear inequality which is tightened by $\Q_i$.
\end{enumerate}
\item[(3)] Output $\bm{S}$.
\end{enumerate}

\underline{\bf Decoding:}

\noindent
{\em Input:} a polytope $\Q \in \R^d$ which is defined by $n$ constraints (linear inequalities) and a sequence $\bm{S}$ of $d$ constraints.\\
{\em Output:} a sequence of vertices $\bm{x}_0,\ldots,\bm{x}_{d}\in \Q$.
\begin{enumerate}
\item[(1)] Initialize $\Q_0 = \Q$, $\bm{x}_0 = \bm{p}(\Q_0)$.
\item[(2)]  For $i=1,\ldots,d$:
\begin{enumerate}
\item[(2.1)] Set $\Q_i$ to be the facet of $\Q_{i-1}$ which is defined by tightening the $i$'th constraint in $\bm{S}$.
\item[2.2] Set $\bm{x}_i = \bm{p}(\Q_i)$.
\end{enumerate}
\item[(3)] Output $\bm{x}_0,\ldots \bm{x}_d$.
\end{enumerate}
\end{tcolorbox}
\caption{ The encoding procedure receives $\Q$ and $a\in \Q$ as inputs 
	and outputs a sequence $\bm{S}$ of $d$ out of the $n$ linear inequalities used to define $\Q$. 
	The decoding procedure receives $\Q$ and $\bm{S}$ as inputs and output 
	a sequence~$\bm{x}_{0},\ldots,\bm{x}_{d}$ of vertices of $\Q$ such that $a \in \conv(\{\bm{x}_{0},\ldots,\bm{x}_{d}\})$.
	Since there are at most $n^{d}$ sequences $\bm{S}$ and since every point $a\in\Q$ is contained in one of the decoded subpolytopes
	$\conv(\{\bm{x}_{0},\ldots,\bm{x}_{d}\})$, it follows that $\Q$ can be covered by $n^{d}$ such subpolytopes as required.   }
\label{fig:encode}
\end{figure}

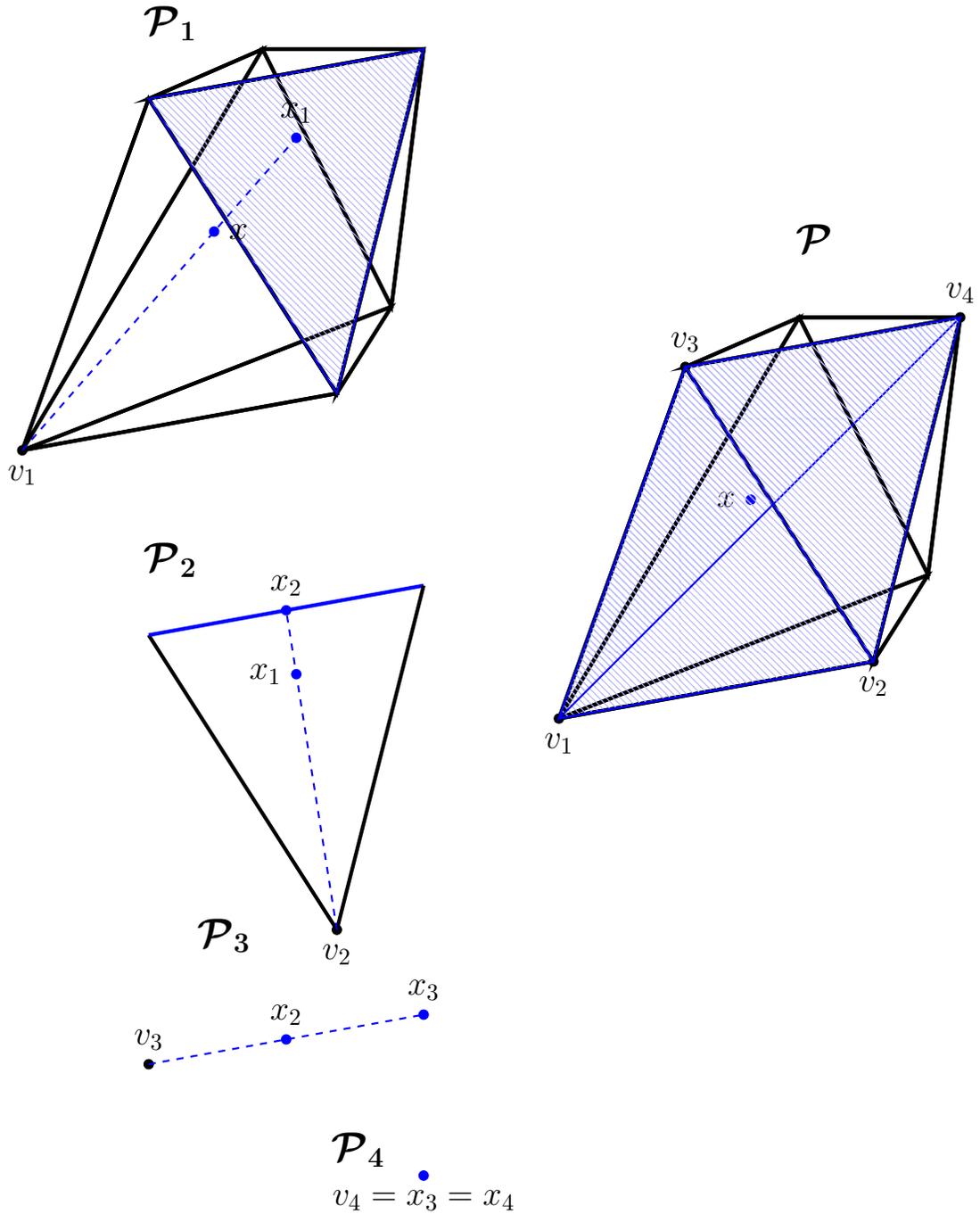
\begin{figure}
\begin{center}
\begin{tikzpicture}[v/.style = {inner sep = 0.5mm, draw, fill, black, circle}, a/.style = {inner sep = 0.5mm, draw, fill, blue, circle}, e/.style = {blue, draw, dashed, thick}, f/.style = {thick}, f1/.style = {ultra thick}, scale=0.8]

	\begin{scope}[scale = 0.6]
        \coordinate (v1) at (0,0,0) {};
        \coordinate (v2) at (9,1,-2) {};
        \coordinate (v3) at (2,9,-5) {};
        \coordinate (v4) at (8,1,-9) {};
        \coordinate (v5) at (4,9,-9) {};
        \coordinate (v8) at (9,9,-9) {};

        \coordinate (label) at (6,13,-5) {};
        \end{scope}

        \coordinate (a3) at (v8) {};
        \coordinate (a2) at ($(v3)!0.5!(a3)$) {};
        \coordinate (a1) at ($(v2)!0.8!(a2)$) {};
        \coordinate (a0) at ($(v1)!0.7!(a1)$) {};
        

        \node[v, label = -90:{\large $ v_1$}] (w1) at (v1) {};
        \node[v, label = -90:{\large $ v_2$}] (w2) at (v2) {};
        \node[v, label = 90:{\large $ v_3$}] (w3) at (v3) {};
        \node[v, label = 90:{\large $ v_4$}] (w4) at (v8) {};
        \node[a, label = 180:{\large $x$}] (c0) at (a0) {};
        
        \draw[f1] (v1) -- (v2) -- (v3) -- cycle;
        \draw[f1] (v1) -- (v3) -- (v5) -- cycle;
        \draw[f1] (v1) -- (v5) -- (v4) -- cycle;
        \draw[f1] (v1) -- (v4) -- (v2) -- cycle;
        \draw[f1] (v8) -- (v2) -- (v3) -- cycle;
        \draw[f1] (v8) -- (v3) -- (v5) -- cycle;
        \draw[f1] (v8) -- (v5) -- (v4) -- cycle;
        \draw[f1] (v8) -- (v4) -- (v2) -- cycle;

        \draw[f, blue, pattern = north west lines, pattern color = blue!30] (v1) -- (v2) -- (v3) -- cycle;
        \draw[f, blue, pattern = north west lines, pattern color = blue!30] (v1) -- (v3) -- (v8) -- cycle;
        \draw[f, blue, pattern = north west lines, pattern color = blue!30] (v1) -- (v2) -- (v2) -- cycle;
        \draw[f, blue, pattern = north west lines, pattern color = blue!30] (v3) -- (v2) -- (v8) -- cycle;
   
	\coordinate (label0) at (label) {};
	\node[circle] (labelnode0) at (label0) {\Large $\bm{\mathcal{P}}$};

	\foreach \n/\vrt in {1/v1, 2/v2, 3/v3, 4/v4, 5/v5, 8/v8}{
		\coordinate (copy1v-\n) at ($ (\vrt) - (10,-5,0)  $) {};
	}
	\foreach \n/\vrt in {0/a0, 1/a1, 2/a2, 3/a3}{
		\coordinate (copy1a-\n) at ($ (\vrt) - (10,-5,0)  $) {};
	}
	\node[v, label=-90:{\large $v_1$}] (copy1b-1) at (copy1v-1) {};
        \draw[f1] (copy1v-1) -- (copy1v-2) -- (copy1v-3) -- cycle;
        \draw[f1] (copy1v-1) -- (copy1v-3) -- (copy1v-5) -- cycle;
        \draw[f1] (copy1v-1) -- (copy1v-5) -- (copy1v-4) -- cycle;
        \draw[f1] (copy1v-1) -- (copy1v-4) -- (copy1v-2) -- cycle;
        \draw[f1] (copy1v-8) -- (copy1v-2) -- (copy1v-3) -- cycle;
        \draw[f1] (copy1v-8) -- (copy1v-3) -- (copy1v-5) -- cycle;
        \draw[f1] (copy1v-8) -- (copy1v-5) -- (copy1v-4) -- cycle;
        \draw[f1] (copy1v-8) -- (copy1v-4) -- (copy1v-2) -- cycle;
        \draw[f, blue, pattern = north west lines, pattern color = blue!30] (copy1v-3) -- (copy1v-2) -- (copy1v-8) -- cycle;
	\node[a, label = 0:{\large $x$}] (copy1d-0) at (copy1a-0) {};
	\node[a, label = 90:{\large $x_1$}] (copy1d-1) at (copy1a-1) {};
        \path[e] (copy1v-1) -- (copy1a-1);

	\coordinate (label1) at ($ (label0) - (12,-4,0)  $) {};
	\node[circle] (labelnode1) at (label1) {\Large $\bm{\mathcal{P}_1}$};

	\foreach \n/\vrt in {3/v3, 2/v2, 8/v8}{
		\coordinate (copy2v-\n) at ($ (\vrt) - (10,5,0)  $) {};
	}
	\foreach \n/\vrt in {0/a0, 1/a1, 2/a2, 3/a3}{
		\coordinate (copy2a-\n) at ($ (\vrt) - (10,5,0)  $) {};
	}
	\node[v, label=-90:{\large $v_2$}] (copy2b-1) at (copy2v-2) {};
        \draw[f1] (copy2v-3) -- (copy2v-2);
        \draw[f1, blue] (copy2v-3) -- (copy2v-8);
        \draw[f1] (copy2v-8) -- (copy2v-2);
	\node[a, label = 180:{\large $x_1$}] (copy2d-0) at (copy2a-1) {};
	\node[a, label = 90:{\large $x_2$}] (copy2d-1) at (copy2a-2) {};
        \path[e] (copy2v-2) -- (copy2a-2);

	\coordinate (label2) at ($ (label0) - (12,6,0)  $) {};
	\node[circle] (labelnode2) at (label2) {\Large $\bm{\mathcal{P}_2}$};

	\foreach \n/\vrt in {3/v3, 8/v8}{
		\coordinate (copy3v-\n) at ($ (\vrt) - (10,13,0)  $) {};
	}
	\foreach \n/\vrt in {0/a0, 1/a1, 2/a2, 3/a3}{
		\coordinate (copy3a-\n) at ($ (\vrt) - (10,13,0)  $) {};
	}
	\node[v, label=90:{\large $v_3$}] (copy3b-1) at (copy3v-3) {};
	\node[a, label = 90:{\large $x_2$}] (copy3d-0) at (copy3a-2) {};
	\node[a, label = 90:{\large $x_3$}] (copy3d-1) at (copy3a-3) {};
        \path[e] (copy3v-3) -- (copy3a-3);

	\coordinate (label3) at ($ (label0) - (11,13,0)  $) {};
	\node[circle] (labelnode3) at (label3) {\Large $\bm{\mathcal{P}_3}$};

	\foreach \n/\vrt in {8/v8}{
		\coordinate (copy4v-\n) at ($ (\vrt) - (10,16,0)  $) {};
	}
	\foreach \n/\vrt in {0/a0, 1/a1, 2/a2, 3/a3}{
		\coordinate (copy4a-\n) at ($ (\vrt) - (10,16,0)  $) {};
	}
	\node[v, blue, label=-90:{\large $v_4 = x_3 = x_4$}] (copy4b-1) at (copy4v-8) {};

	\coordinate (label4) at ($ (label0) - (8.5,17,0)  $) {};
	\node[circle] (labelnode4) at (label4) {\Large $\bm{\mathcal{P}_4}$};

\end{tikzpicture}
\end{center}
\caption{An illustration of bottom vertex triangulation for the polytope $\mathcal{P}$ and a point~$x \in \mathcal{P}$. The process starts by shooting a ray from the bottom vertex, (i.e.\ $v_1$) to $x$. The ray is extended untill it hits one of the faces to the polytope at a point which is denoted by $x_1$. The process is then repeated with the face as a polytope with one fewer dimension. }
\label{fig:bvt}
\end{figure}

\subsection{Convex Set Disjointness}\label{sec:overviewdisj}
\subsubsection*{Upper Bound}
Imagine for simplicity that $d=O(1)$.
Already in this regime, deriving an $o(n)$ upper bound is non-trivial.\footnote{The case of $d=1$ is easy, $d=2$ is more sophisticated, and $d=3$ seems to require a general approach.}
\cite{kane17communication} present a natural protocol based on {\it boosting/multiplicative-weights update rule}
with $\Theta(\log^2 n)$ communication complexity.
Such quadratic dependence is also exhibited by other approaches ({\em e.g.},\ the protocol by \cite{vempala19optimization} which is based on Clarkson's algorithm).
Roughly speaking, this is because these protocols take $\Theta(\log n)$ rounds\footnote{Kane et al.\ prove  that any optimal protocol must have $\tilde\Omega(\log n)$ rounds.} with $\Theta(\log n)$ bits per round. Improving upon this quadratic dependence is already non-trivial.
Our approach is based on two steps.


\paragraph{Step (i): Reducing to a Promise Variant (\Cref{{lem:reduction}}).}
Let $\prob{PromiseCSD}_U$ denote the variant of Convex Set Disjointness in which it is {\it promised} that the inputs $X,Y$ satisfy:
\begin{itemize}
\item[(i)] $\conv(X)\cap\conv(Y)=\emptyset$, or 
\item[(ii)] $X\cap Y\neq \emptyset$.
\end{itemize}
(In particular, the output of the protocol is not restricted in the remaining case when 
$X\cap Y = \emptyset$ and $\conv(X)\cap\conv(Y)\neq\emptyset$).

Clearly, $\prob{PromiseCSD}_U$ can only be easier to decide than $\prob{CSD}_U$.
In the opposite direction, it turns out that it is not much harder.
Specifically, one can reduce to the promise variant
by adding at most $(2n)^{d+2}$ carefully chosen points to the domain.
The idea is to use \Cref{lem:bicaratheodory} which asserts that if $\conv(X)\cap\conv(Y)\neq\emptyset$
then there are $X'\subseteq X$ and $Y'\subseteq Y$
such that $\lvert X'\rvert+\lvert Y'\rvert \leq d+2$ and $\conv(X')\cap\conv(Y')\neq\emptyset$. 
Therefore, for every pair of sets~$X',Y'\subseteq U$ such 
that~$\lvert X'\rvert+\lvert Y'\rvert \leq d+2$, we add to $U$ an auxiliary point in~$\conv(X')\cap \conv(Y')$.
 Then, whenever $\conv(X)\cap\conv(Y)\neq\emptyset$, their intersection
 must contain one of the auxiliary points.
 
We then devise a protocol for $\prob{PromiseCSD}_U$ with
communication complexity 
\[O(d\log d\log n).\] 
This implies the stated upper bound
of $O(d^2\log d\log n)$
on $\prob{CSD}_U$,
since the reduction to the promise variant enlarges
the domain by at most $(2n)^{d+2}$ points.

\begin{figure}

\begin{center}
\begin{tikzpicture}[vb/.style = {inner sep = 0.8mm, draw, fill, blue, circle}, vbnf/.style = {inner sep = 2mm, draw, blue, circle, ultra thick}, vg/.style = {inner sep = 0.8mm, draw, fill, red, circle}, vgnf/.style = {inner sep = 2mm, draw, red, circle, ultra thick}, scale = 0.8]

        \coordinate (c1) at (-6.5,0) {};
        \coordinate (c2) at (6.5,0) {};
        \draw[thin, black] (c1) -- (c2);
        
        \node[vg] (p1) at (6,0) {};
        \node[vb] (p2) at (4.5,0) {};
        \node[vb] (p3) at (4,0) {};
        \node[vb] (p4) at (3.2,0) {};
        \node[vb] (p5) at (2.6,0) {};
        \node[vg] (p6) at (1.1,0) {};
        \node[vg] (p7) at (0,0) {};
        \node[vg] (p9) at (-1.8,0) {};
        \node[vg] (p10) at (-3.5,0) {};
        \node[vg] (p0) at (-4.6,0) {};
        \node[vg] (p13) at (-6,0) {};

        \draw[dashed,ultra thick, shorten >=-0.5cm,shorten <=-0.5cm] (2,-0.1) -- (2,0.1);
        
\end{tikzpicture}
\end{center}

\caption{
The algorithm for the promise variant does not extend to the general case: the figure depicts a case where the convex hull of the red points intersects the convex hull of the blue points. Since the halfspace on the right of the dashed hyperplane contains all the blue points and less than half of the total, the parties will decide to remove all the points to the left of the hyperplane. However, once these points are removed from consideration, the convex hulls of the remaining red and blue points are disjoint. 
}
\label{fig:1Dnopromise}
\end{figure}
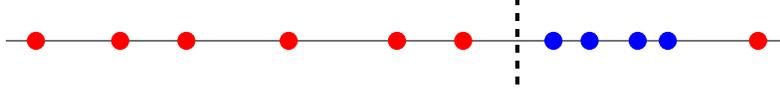

\paragraph{Step (ii): Solving the Promise Variant (\Cref{thm:ub}).}
It remains to explain how $\prob{PromiseCSD}_U$ can be solved
with $\tilde O(d\log n)$ bits of communication.
As a warmup, note that devising a non-trivial protocol for $\prob{PromiseCSD}_U$ is
considerably easier than for $\prob{CSD}_U$: 
indeed, if $\conv(X)\cap\conv(Y)=\emptyset$,
then $X$ and $Y$ can be separated by a hyperplane and one of the two 
halfspaces it defines contains at most~$n/2$ points from $U$. 
This suggests the following approach: Alice and Bob each privately checks if their input
lies in a halfspace which contains at most $n/2$ points from $U$.
If there is no such halfspace then by the above reasoning it must be the case that $\conv(X)\cap\conv(Y)\neq\emptyset$ and the protocol terminates.
Else, they can agree on such a halfspace using $O(d\log n)$ bits 
and remove all domain points outside this halfspace (the bound on the number of bits is because there are~$n^{O(d)}$ halfspaces
up to equivalences\footnote{Two halfspaces are equivalent if they have the same intersection with $U$.}).
Alice and Bob can iteratively proceed in this manner and in every step remove at least half of the (remaining) points 
while maintaining that all points in $X\cap Y \subseteq U$ are never being removed.
The implied protocol has a total of $O(\log n)$ rounds, and in each round $O(d\log n)$ bits are communicated. 
Thus, the total number of bits is $O(d\log^2n)$ (which is $\log n$ factor away from the stated bound). 

Our final protocol uses a similar recursive approach, but {\it transmits only $O(d\log d)$ bits in each round}.
This is achieved by using Halfspace Containers (\Cref{thm:onesidedcover}). 
Specifically, instead of finding a halfspace which contains the entire input
of one of the players, they find an $\eps$-container for this halfspace with $\eps=1/4$.
This allows to reduce the domain size by
a factor of $1/2 + 1/4=3/4$ in each round and, by \Cref{thm:onesidedcover}, requires only $d\log d$ bits per round.
The proof of \Cref{thm:onesidedcover} is sketched in \Cref{sec:containeroverview}.

One may be tempted to try a similar approach for the non-promise variant.
However, note that points in $\conv(X)\cap\conv(Y)$ that are not in $X\cap Y$ may be removed by the protocol. Indeed, 
\Cref{fig:1Dnopromise} depicts a situation where the protocol starts with sets $X,Y$ with $\conv(X)\cap\conv(Y)\neq\emptyset$ and removes some of the points in $U$ to obtain a domain $U'$ in which $\conv(X \cap U')\cap\conv(Y \cap U') = \emptyset$. This shows that without the promise, this approach may fail.

\subsubsection*{Lower Bound}

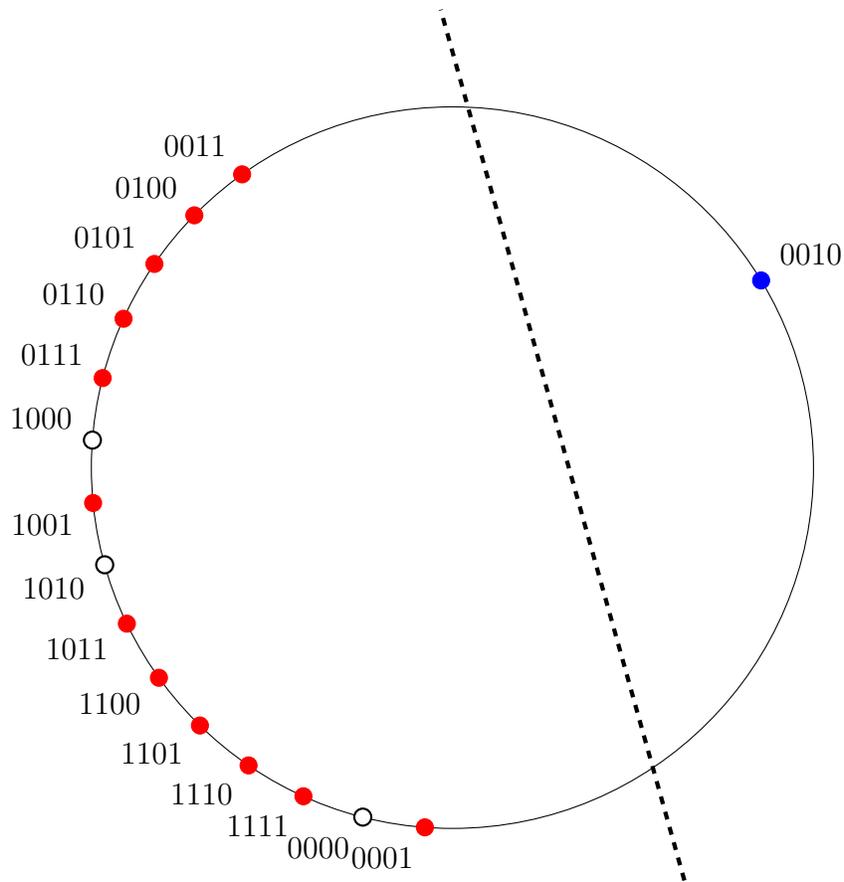
\begin{figure}
\begin{center}
\begin{tikzpicture}[vb/.style = {inner sep = 0.8mm, draw, fill, blue, circle}, vg/.style = {inner sep = 0.8mm, draw, fill, red, circle}, vbl/.style = {inner sep = 0.8mm, draw, black, fill=white, thick, circle}, scale = 0.8]
\draw (0,0) circle (6);
%
\foreach \n/\m in {0001/17, 0011/3,0100/4,0101/5,0110/6,0111/7,1001/9,1011/11,1100/12,1101/13,1110/14,1111/15} {
  \coordinate (v\m) at (10 * \m + 5.625 + 90:  6) {};
   \node [vg, label = {10 * \m + 5.625 + 90}:\n] (u\m) at (v\m) {};
    }
%
\foreach \n/\m in {0000/16, 1000/8,1010/10} {
  \coordinate (v\m) at (10 * \m + 5.625 + 90:  6) {};
   \node [vbl, label = {10 * \m + 5.625 + 90}:\n] (u\m) at (v\m) {};
    }
    
\foreach \n/\m in {0010/2} {
  \coordinate (v\m) at (10 * \m + 11.25:  6) {};
   \node [vb, label = {10 * \m + 11.25}:\n] (u\m) at (v\m) {};
    }

       \coordinate (c12) at ($(v17)!0.5!(v2)$) {};
       \coordinate (c23) at ($(v2)!0.5!(v3)$) {};
        \draw[dashed, ultra thick, shorten >=-3cm,shorten <=-4.5cm] (c12) -- (c23);
    
\end{tikzpicture}
\end{center}
\caption{A depiction of the reduction from Set Disjointness on $4$ bits to Convex Set Disjointness on $16$ points. Alice's input in Set Disjointness is $0010$ while Bob's input is~$0101$. The domain $U$ of Convex Set Disjointness has $16$ equally spaced points on the unit circle (not the case in the figure above, to emphasize the dashed separating hyperplane). Alice's input is mapped to a single point, in this case, the point $0010$. 
As every point in this construction can be separated by a line (in this figure, {\em e.g.}, the blue point), it follows that the convex hulls of Alice's and Bob's points are disjoint if and only if Alice's input is mapped to a point which is not in the set of points Bob's input is mapped to, which, in turn, happens if and only if the inputs of Alice and Bob for Set Disjointness were disjoint. }
\label{fig:4}
\end{figure}

We prove a stronger lower bound then the one stated in \Cref{thm:mainlb}.
In particular, in \Cref{sec:LB} we derive an~$\Omega(d\log(n/d))$ lower bound  which applies even to the promise variant.

The first part in the lower bound is a reduction from Set Disjointness on $\log m$ bits 
to planar convex set disjointness with $m$ points.
This achieved by fixing $m$ points in a convex position, say on the unit circle,
and identifying each $\log m$ bit-string $\bm{z}$ with one of the $m$ points.
Thus, for a bit-string $\bm{z}$, let $v_{\bm{z}}$ denote the corresponding point on the unit circle.
Next, given inputs~$\bm{x},\bm{y}\in\{0,1\}^{\log m}$, Alice transform her input to the singleton set $\{v_{\bm{x}}\}$,
whereas Bob transform his input to the set $\bigl\{v_{\bm{z}} \vert (\exists i) : \bm{y}(i)=\bm{z}(i)=1\bigr\}$.
Note that Alice's point is in Bob's set if and only if $\bm{x}\cap \bm{y}\neq\emptyset$.
Moreover, since the $m$ points are in convex position, Alice's point is in Bob's set
if and only if it can not be separated from it by a hyperplane; 
i.e.\ if and only if their convex hulls intersect.
This establishes a reduction from Set Disjointness on $\log m$ bits
to (promise) Convex Set Disjointness on $m$ points in $\mathbb{R}^2$.
See \Cref{fig:4} for an illustration of this construction.

The second part of the lower bound is to lift the planar construction to higher dimensions in a way that preserves the logic of the reduction:
we take $d$ orthogonal copies~$U_1,\ldots, U_d$ of the planar construction,
each of size $n/d$ and place them such that the following holds.
Let~$X,Y\subseteq \bigcup_i U_i$ be possible inputs for Alice and Bob and let $X_i = X\cap U_i, Y_i= Y\cap U_i$. 
Then,
\[\Bigl((\forall i): \conv(X_i)\cap\conv(Y_i) = \emptyset\Bigr) \implies  \conv(X)\cap\conv(Y) = \emptyset.\]
Specifically, the $U_i$'s are placed such that if $\vec n_i$ is the normal of a hyperplane
separating $X_i$ and $Y_i$, then the vector $\vec n = \sum_i \vec n_i$ is the normal to a hyperplane
that separates $X$ and $Y$.

\subsection{Learning Halfspaces}\label{sec:overviewlearning}
The bounds for Learning Halfspaces follow from the corresponding bounds for CSD.

The lower bound utilizes the lower bound for the promise variant of CSD.
The promise plays a key role in enabling the lower bound to apply also to improper protocols.
Indeed, it is not hard to see that an improper learning protocol can be used to decide the promise variant.
The argument is straightforward, and we refer the reader to \Cref{sec:hslb} for the complete short proof.

The upper bound is based on the $\tilde O(d\log n)$ protocol 
for the promise variant. Specifically, it exploits its following property:
in the case when the convex hulls of $X,Y$ are disjoint,
the protocol returns a certificate in the form of a function $f:U\to\{\pm 1\}$
such $f(u) = +1$ for every $u\in Y$ and $f(u)=-1$ for every $u\in X$ (see \Cref{lem:promise}).
This immediately yields a learning protocol in the case when Alice only has
negative examples and Bob only has positive examples.
The case where both Alice and Bob may have mixed examples is
more subtle, but the protocol and analysis remain rather simple.
We refer the reader to \Cref{sec:hsub} for the complete proof.

%
%
%
%
%

%% file: epscoversbetter2.tex

\section{A Container Lemma for Halfspaces} \label{sec:epscovers}

We establish here the existence of a small set of containers for halfspaces in $\R^d$.
\begin{theorem*}[Theorem~\ref{thm:onesidedcover} restated]
Let $U\subseteq \R^d$. Then, for every $\eps > 0$ there is a set of $\eps$-containers for $\hs(U)$ of size $(d/\eps)^{O(d)}$.
\end{theorem*}
This section is organized as follows:
\Cref{sec:VCpreliminaries} contains some basic facts from VC theory.
In \Cref{sec:haussler} we discuss how this result relates with a classical result by Haussler which has a similar flavour~\cite{haussler1995sphere}.
Finally, a complete proof of \Cref{thm:onesidedcover} is given in \Cref{sec:coverproof}.

\subsection{Preliminaries from VC theory.}\label{sec:VCpreliminaries}
We will use two basic results from VC theory.
Recall that the {\it VC dimension} of a family~$\F\subseteq \mathbbm{2}^X$
is the size of the largest $Y\subseteq X$
such that $\{F\cap Y : F\in\F\}= \mathbbm{2}^Y$.
An {\it $\eps$-net} for $\F$ is a set $N\subseteq X$ such that $N\cap F\neq\emptyset$ for all $F\in \F$ with $\lvert F\rvert \geq \epsilon\lvert X\rvert$.  
A useful property of families
with small VC-dimension is that they have small $\eps$-nets. 
\begin{theorem}[$\eps$-net Theorem]\citep{haussler1986epsilon,vapnik2015uniform}\label{thm:vc}
Let $\F\subseteq \mathbbm{2}^{X}$ be a family with VC dimension $d$ and let~$\eps>0$.
Then, there exists an $\eps$-net for~$\F$ of size $O \left(\frac{d\log(1/\eps)}{\eps} \right)$.
\end{theorem}

We will also use the following lemma which bounds the growth in the VC dimension
under set operations:
\begin{lemma}[VC of $k$-fold compositions]\citep{BEHW89}\label{lem:vcofcomp}
Let $\F_1\ldots \F_k$ be a sequence of families with VC dimension at most $d$, 
and let $\star_1\ldots \star_{k-1}$ be a sequence of binary operations on sets
(e.g.\ $\star_1=\cap, \star_2=\cup,\star_3= \Delta$, and so forth). 
Set \[\F^{\star k} = \Bigl\{ F_1\star_1 (F_2\star_2\ldots(F_{k-1}\star_{k-1} 
F_k)) : F_i\in \F_i \Bigr\}.\]
Then, the VC dimension of $\F^{\star k}$ is at most $O(kd\log d)$.
\end{lemma}
This Lemma allows to use the VC dimension of $\F$ to bound the VC dimension of more complex families, {\em e.g.}, 
\[\Bigl\{\bigl(F_1\setminus (\cap_{i=2}^{100}F_i)\bigr) \cup F_{101} : F_i \in \F\Bigr\}.\]

\subsection{Comparison with Haussler's Packing Lemma}\label{sec:haussler}
Theorem~\ref{thm:onesidedcover} is closely related to a result by \cite{haussler1995sphere}, 
which asserts that every family $\F\subseteq \mathbbm{2}^X$ with VC dimension $d$
(e.g. $d-1$ dimensional halfspaces) has an $\eps$-cover of size roughly~$(1/\eps)^d$,
where an $\eps$-cover is a family $\C$ such that for every $F\in \F$ there is $C\in \C$
such that $\lvert F\Delta C\rvert \leq \eps\lvert X\rvert$ (see \Cref{sec:containers}).

We note that unlike Haussler's result,
Theorem~\ref{thm:onesidedcover} does not extend to arbitrary VC classes
(below is a counterexample with VC dimension 2). 
This is also reflected in our proof of Theorem~\ref{thm:onesidedcover}
which exploits the dual variant of  Carath\'eodory's Theorem (\Cref{lem:dualcarath}), which does not extend to arbitrary VC classes.

\paragraph{Example.}
Consider a projective plane $P$ of order $n$ with $N=n^2+n+1$ points and $N$ lines.
In particular the following holds: (i) for every pair of points there is a unique line containing them, 
(ii) every pair of lines intersects in one point, 
(iii) every line contains exactly $n$ points,
(iv) and every point is contained in exactly $n$ lines.

Let $\F$ be the family 
\[\{L : L \text{ is a line in $P$}\}.\]
One can verify that $\F$ has VC dimension 2.
Set $\eps=1/4$. 
Since each line contains~$n=O(\sqrt{N})$ points,
then for a sufficiently large $N$,
the existence of a set of $\eps$-containers for $\F$ of size~$t$ amounts to
the following statement:
\begin{center}
There exist $t$ sets of size at most $N/3$ each, 
such that every line in $P$ is contained in at least one of them.
\end{center}
Therefore, by averaging, one of these $t$ sets 
contains at least~$N/t$ lines $L_1,\L_2,\ldots L_{N/t}$. 
Denote such a set by $C$.
Assume towards contradiction that $t$ depends only on $\eps=1/4$ and $d=2$,
and in particular that $t\leq N/n = \theta(n)$.
Now, since every two lines intersect in one point
it follows that  
\begin{align*}
\lvert\cup_{i=1}^{N/t} L_i \rvert &\geq n + (n-1) + \ldots  + 1   \tag{because $\lvert L_i \setminus \cup_{j< i}L_j\rvert \geq n- (i-1)$}\\
			                    &\geq n^2/2,
\end{align*}
where in the first inequality we used that $N/t\geq n$.
Thus, since $C$ contains this union:
\[ n^2/2 \leq \lvert C\rvert \leq N/3 = (n^2 + n + 1)/3,\]
which is a contradiction when $n$ is sufficiently large.

\subsection{Proof of Container Lemma (Theorem~\ref{thm:onesidedcover})}\label{sec:coverproof}


\paragraph{The superset $\C'$.}  Let $\C' =\{U\setminus \bigl(\cap_{i=1}^d H_i) : H_i\in \hs_d \}$. It is easy to see that $\C'\supseteq \hs(U)$, 
and therefore it is an $\eps$-cover for $\hs(U)$, for every $\eps$. 
However $\C'$ is a much larger set than we can afford. 
The final cover $\C$ will be a carefully selected subfamily of $\C'$.

To select the subset $\C \subseteq \C'$, we use the following observation that provides a criteria to certify that $\C$ is a set of $\eps$-containers for~$\hs_d$: 
it suffices to show that for every $H \in \hs_d$ there is $C \in \C$ such that $C$ is an $\eps$-container for $F$.
Here, for any $C,F \subseteq \mathbbm{2}^X$, we say that $C$ is an $\eps$-container for $F$ if $F\subseteq C$, and $\lvert C \setminus F\rvert\leq\eps\lvert X\rvert$. 

\begin{observation}\label{obs:epsnet}
Let $\F,\C \subseteq \mathbbm{2}^X$.
Let $V$ be an $\eps$-net for $\{C' \setminus F' : C' \in \C, F' \in \F\}$.
Let~$C\in \C$ and $F\in \F$ such that
\begin{enumerate}
\item $F\subseteq C$ and 
\item $C\cap V = F\cap V$.
\end{enumerate}
Then, $C$ is an $\eps$-container for $F$. (Namely, 
$F\subseteq C$, and $\lvert C \setminus F\rvert\leq\eps\lvert X\rvert$).
\end{observation}

\begin{proof}
Given items 1 in the observation, 
it remains to show that $\lvert C\setminus F\rvert \leq \eps\lvert X\rvert$.
This follows by the second item,
which implies that $\emptyset=(C\cap V) \setminus (F\cap V)=(C\setminus F)\cap V $,
and since $V$ is an $\eps$-net for $\{C' \setminus F' : C' \in \C, F' \in \F\}$. 
We get that $\lvert C\setminus F\rvert \leq \eps\lvert X\rvert$, as required.
\end{proof}

\paragraph{The $\eps$-net $V$.} Our selection of $\C\subseteq \C'$ hinges on \Cref{obs:epsnet}, 
and therefore we use an $\eps$-net $V$ for the family $\C'' = \{C' \setminus H' : C' \in \C', H' \in \hs_d\}$ of size \[\lvert V\rvert= O\left(\frac{d^2\log d \log(1/\eps)} {\eps}\right).\] 
(Note, in particular, that $V$ is an $\epsilon$-net for every subfamily of $\C''$).
The bound on $\lvert V\rvert$ follows from \Cref{thm:vc} because the VC dimension of $\C''$ is $O(d^2\log d)$. 
This bound on the VC dimension of $\C''$ follows because the VC dimension of $\hs_d$ is $d+1$, thus, due to Lemma~\ref{lem:vcofcomp}, the VC dimension of $\C'$ and $\C''$ is $O(d^2\log d)$.

\paragraph{The family of containers $\C$.} Next we construct $\C$. 
The construction is based on an encoding-decoding scheme: given a halfspace~$H\in \hs(U)$, the scheme encodes $H$ into a bit-string $\bm{b} = \bm{b}(H)$ of length $t=O(d\log\lvert V\rvert)$. The bit-string $\bm{b}$ is then decoded to a set~$C = C(\bm{b})\in \C'$ satisfying the two items in Observation~\ref{obs:epsnet} with respect to $V$ -- and therefore $C$ is an $\eps$-container of~$H$.
The upper bound on the length~$t$ of $\bm{b}$ implies that the collection $\{ C(\bm{b}) :\; \bm{b} \in \{0,1\}^t\} \subseteq \C'$ is a set of $\eps$-containers for $\hs(U)$ of size~$2^t={O(d\log \lvert V\rvert)} = \lvert V\rvert^{O(d)} = (d/\eps)^{O(d)}$.

Let $H\in \hs(U)$. Let $a\in \R^d$, $\| a\|_\infty \leq 1$ and $b\in \R, \lvert b\rvert \leq 1$ be such that 
\[H = \{u\in U : \ip{a,u} < b\}.\]
{Moreover, since $U$ is finite, we may assume without loss of generality that there exists a universal\footnote{I.e. that depends only on $U$.} small constant $\eps>0$ such that $\ip{a,u} < b-\eps$ for every $u\in H$ and $\ip{a,u} > b+\eps$ for every $u\in U\setminus H$.}

The rest of the proof is devoted to constructing an $\eps$-container $C$ for $H$ by first constructing~$\bm{b} = \bm{b}(H)$ and then $C = C(\bm{b})$.

\paragraph{The auxiliary polytope $\P$.}
The definition of $\bm{b}(H)$ uses a polytope $\P$ that we define next.
Recall that $V\subseteq U$ is an $\eps$-net for $\C''=\{C' \setminus H' : C'\in \C', H'\in \hs_d\}$.
Let $V^-=V\cap H=\{v\in V : \ip{a,v} < b\}, V^+ = V\setminus H =  \{v\in V : \ip{a,v} \geq b\}$.
Define~$\P\subseteq \R^{d+1}$:
\[ \P = \Bigl\{(\alpha,\beta)\in\R^d\times \R ~\Big\vert \; \bigl(\|(\alpha,\beta)\|_\infty\leq 1\bigr) \land \bigl(\forall v\in V^+: \ip{\alpha, v} \geq \beta +\eps\bigr) \land \bigl(\forall v\in V^-: \ip{\alpha, v } \leq \beta -\eps\bigr)  \Bigr\}.\]
{Observe that~$\P$ contains a representation $(\alpha,\beta)$ for each halfspace $H'= \{u\in U : \ip{\alpha,u} < \beta\}$ 
such that $H'\cap V = H\cap V = V^-$, and only such representations.}
The constraint $\|(\alpha,\beta)\|_\infty\leq 1$ ensures that $\P\subseteq \R^{d+1}$ is bounded,
a property which will be enable us to apply \Cref{lem:dualcarath} to $\P$.
Note that $\P$ is a closed polytope which is defined by $\lvert V\rvert + 2(d+1)$ linear inequalities 
(the constraint $\|(\alpha,\beta)\|_\infty\leq 1$ amounts to $2(d+1)$ linear inequalities).
Moreover, note that $\P$ is non-empty, since $(a,b)\in\P$ (see \autoref{fig:auxpoly}).


\paragraph{The encoding $\bm{b}(H)$.}
The bit-string $\bm{b} = \bm{b}(H)$ encodes the polytope $\P$, as well as the names of $d+2$ vertices $\bm{x}_0,\ldots,\bm{x}_{d+1}$ of $\P$ such that $(a,b) \in \conv(\{\bm{x}_0,\ldots,\bm{x}_{d+1}\})$ is in their convex hull (the existence of such vertices is promised by the Carath\'eodory's Theorem).

The polytope $\P$ can be encoded using $O(d\log d)$ bits, as $\P$ is determined by $V^-=H\cap V\in\hs(V)$, and $V^-$ can be described using 
$\log\lvert \hs(V)\rvert \leq d\log \lvert V\rvert + 1 = O(d\log d)$
bits, where the first inequality is because $\lvert \hs(V)\rvert \leq 2\lvert V\rvert^d$ (see, {\em e.g.},~\citep{Gartner94vapnik}).

The points $\bm{x}_0,\ldots,\bm{x}_{d+1}$ can be naively conveyed using $d^2 \log d$ bits\footnote{To see that, observe that the number of vertices in $\P$ is $O({\lvert V\rvert + 2(d+1) \choose d+1}) = \exp(d\log d)$, because $\P$ is defined by $\lvert V\rvert + 2(d+1)$ constraints, and each vertex is determined by $d+1$ constraints. Therefore, each vertex can be described using $O(d\log d)$ bits, and $d+2$ vertices can be represented by $O(d^2\log d)$ bits.}. To obtain a more compressed representation of these points, we use the dual version of Carath\'eodory Theorem (\Cref{lem:dualcarath}). Since $\P\subseteq \R^{d+1}$ is defined as the intersection of $\lvert V\rvert + 2(d+1)$ halfspaces, \Cref{lem:dualcarath} shows such vertices $\bm{x}_0,\ldots,\bm{x}_{d+1}$ can be represented using $\log (\lvert V\rvert + 2(d+1))^{d+2} = O(d \log d)$ bits.

\paragraph{The decoding $C(\bm{b})$.}
The next lemma shows how an $\eps$-container $C=C(\bm{b})$ for $H$ can be derived from $\bm{b}$, thus concluding the proof of \Cref{thm:onesidedcover}.

\begin{lemma}\label{lem:polytope}
Let $H = \{u\in U : \ip{a,u} < b\}$ as above.
Let $(\alpha_0,\beta_0),\ldots,(\alpha_{d+1},\beta_{d+1})$ be vertices of $\P$ such that $(a,b)\in\conv(\{(\alpha_i,\beta_i)\})$.
Then, the set $C=U\setminus (\bigcap_{i=1}^{d+2} H_i)$, where $H_i = \{x: \ip{\alpha_i,x} \geq \beta _i\}$,
satisfies the two items in Observation~\ref{obs:epsnet} with respect to $H$.
\end{lemma}

\begin{proof}

{(i) $H\subseteq C$:} 
let $u\in H$. Therefore, $u\in U$ and $\ip{a,u} < b$. 
Now, since $(a,b)$ is a convex combination of the $(\alpha_i,\beta_i)$'s, 
it must be the case that $\ip{\alpha_i,u} < \beta_i$ for some $i \in \{0,\ldots, d+1\}$, {\em i.e.}, that $u\notin H_i$. 
The reason is that we can write $a = \sum_{i=0}^{d+1} \gamma_i \alpha_i$ and  $b = \sum_{i=0}^{d+1} \gamma_i \beta_i$ where $\gamma_i \in [0,1]$. Thus, if $\ip{\alpha_i,u} \geq \beta_i$ for all $i \in \{0,\ldots, d+1\}$, then $\ip{a,u} = \ip{\sum_{i=0}^{d+1} \gamma_i \alpha_i,u} = \sum_{i=0}^{d+1} \gamma_i \ip{\alpha_i,u} \geq \sum_{i=0}^{d+1} \gamma_i \beta_i = b$, contradicting the fact that $\ip{a,u} < b$.
Since there exists $i \in \{0,\ldots, d+1\}$ such that $u\notin H_i$, we get $u\notin \bigcap_i H_i$. This implies $u\in C$, as required.

(ii) {$C\cap V = H\cap V$:} 
For every $i \in \{0,\ldots, d+1\}$, since $(\alpha_i,\beta_i)\in \P$, it follows that $V \setminus H_i = \bar H_i \cap V = H \cap V= V^-$.
This implies $H \cap V = V^{-} = V \setminus (\bigcap_{i=1}^{m} H_i) = C \cap V$, as required.
\end{proof}

\subsection{Proof of Dual Carath\'eodory Theorem (\Cref{lem:dualcarath})}


\paragraph{The Encoding-Decoding Procedure.}

Let $\Q \subseteq \R^d$ be a polytope which is defined by $n$ linear inequalities and let $\bm{a} \in \Q$. 
	The proof boils down to an encoding and encoding procedures which are based on bottom vertex
	triangulation~\citep{Clarkson88queries,Goodman04handbook} and are described in \Cref{fig:encode}.

The encoding procedure receives $\Q$ and $a\in \Q$ as inputs 
	and outputs a sequence $\bm{S}$ of $d$ out of the $n$ linear inequalities used to define $\Q$. 
	The decoding procedure receives $\Q$ and $\bm{S}$ as inputs and output 
	a sequence~$\bm{x}_{0},\ldots,\bm{x}_{d}$ of vertices of $\Q$ such that $a \in \conv(\{\bm{x}_{0},\ldots,\bm{x}_{d}\})$. 
	That is, $\bm{S}$ encodes a subpolytope defined by~$d+1$ vertices that contains $a$.
	Since there are at most $n^{d}$ such sequences $\bm{S}$ and since every point $a\in\Q$ is contained in one of the encoded subpolytopes, 
	this will imply that $\Q$ can be covered by $n^{d}$ such subpolytopes as required.  

We use the following convention:
	for every polytope $\Q'$, fix a pivot vertex $\bm{p}(\Q')\in \Q'$
	(for example, $\bm{p}(\Q')$ can be the bottom vertex in $\Q$, or the smallest vertex with respect to the lexicographical order, etcetera).
	Also, let $\dim(\Q')$ denote the dimension of~$\Q'$
	 ({\em i.e.}, the dimension of the affine span\footnote{Recall that the affine span of a set $A$ 
	 is the minimal affine subspace that contains $A$.} of $Q$).

\paragraph{Analysis.}
The description of the encoding and decoding procedures appears in \Cref{fig:encode}.
We finish the proof by showing that $\bm{a}\in\conv(\{\bm{x}_0,\ldots,\bm{x}_{d}\})$. 
This follows by induction on $\dim(\Q)$: the base case of $\dim(\Q)=0$ is trivial.
For the induction step, assume that the claim holds 
for every polytope of dimension strictly less than~$k$, and prove the claim for~$\dim(\Q)=k$:
by construction, $\bm{a}$ is a convex combination of $\bm{x}_0$ and $\bm{a}_1$.
Since~$\dim(\Q_1)=k-1$, by the induction hypothesis, $\bm{a}_1$ is in the convex hull of $\bm{x}_1\ldots \bm{x}_d$. 
This implies that $a$ is in the convex hull of $\bm{x}_0\ldots \bm{x}_d$, as required.

%% file: prelim.tex

\section{Communication Complexity Proofs}\label{sec:proofs}

This section is organized as follows.
In \Cref{sec:prelim} we formally define the communication problems
discussed in this paper and survey some elementary tools from 
communication complexity.
In \Cref{sec:UB,sec:LB} we prove \Cref{thm:mainub,thm:mainlb}.

\subsection{Preliminaries} \label{sec:prelim}

%
%


{We use capital letters to denotes sets ({\em e.g.}, $X,Y,U$). We denote by calligraphic capital letters families of sets ({\em e.g.}, $\C,\F$). We use bold small letters to denote vectors ({\em e.g.}, $\bm{x},\bm{y}$). We sometimes write $\bm{x}^{(k)}$ to stress that the vector $\bm{x}$ consists of $k$ coordinates, numbered $1$ to $k$. If $\bm{x}$ is a vector, we denote by~$x_i$ the~$i^{\text{th}}$ coordinate in $\bm{x}$.} 

\subsubsection*{Communication complexity}
We use standard notation and terminology from Yao's communication complexity model \citep{Yao79},
and refer the reader to \citep{Kushilevitz97book} for a textbook introduction. 
{For a (possibly partial) function $f$, we denote by $D(f)$ the deterministic communication complexity of $f$, and by $R_\epsilon(f)$ the randomized communication complexity of $f$ with error probability $\epsilon \geq 0$. We set $R(f) = R_{1/3}(f)$.  }

\begin{definition}[$\prob{DISJ}_n$] The disjointness function $\prob{DISJ}_n : \{0,1\}^n \times \{0,1\}^n \to \{0,1\}$ is defined as:
\[\prob{DISJ}_n(\bm{x},\bm{y}) = \begin{cases}
0 &, \exists i\colon x_i = y_i = 1\\
1&, \text{otherwise}.\\
\end{cases}\]
\end{definition}

\begin{definition}[$\prob{AND}_k$] For a function $f: \mathcal{X} \times \mathcal{Y} \to \{0,1\}$, the function $\prob{AND}_k\circ f : \mathcal{X}^k \times \mathcal{Y}^k \to \{0,1\}$ is defined as:
\[\prob{AND}_k \circ f(\bm{x}^{(k)}, \bm{y}^{(k)}) = \bigwedge_{i=1}^k f(x_i, y_i).\]
\end{definition}

\subsubsection*{Convex set disjointness}

\begin{definition}[$\prob{CSD}_U$] 
Let $U\subseteq \R^d$ be a finite set. The convex set disjointness function $\prob{CSD}_U(X,Y) : \mathbbm{2}^U \times \mathbbm{2}^U \to \{0,1\}$ is defined as:
\[\prob{CSD}_U(X,Y) = \begin{cases}
0 &, \conv(X) \cap \conv(Y) \neq \phi\\
1 &, \text{otherwise}.\\
\end{cases}\]
\end{definition}

\begin{definition}[$\prob{PromiseCSD}_U$] 
Let $U\subseteq \R^d$ be a finite set.
The partial function~$\prob{PromiseCSD}_U(X,Y) : \mathbbm{2}^U \times \mathbbm{2}^U \to \{0,1\}$ is defined as:
\[\prob{PromiseCSD}_U(X,Y) = \begin{cases}
0 &, X \cap Y \neq \phi\\
1 &, \conv(X) \cap \conv(Y) = \phi.\\
\end{cases}\]
\end{definition}

\subsubsection*{Learning halfspaces}

{Fix a finite domain $U\subseteq \R^n$.
An example is a pair $(\bm{x},y)\in U\times\{\pm 1\}$.
An  example $(\bm{x},y)$ is called a positive (negative) example if $y=+1$ ($y=-1$).
A set of examples $S\subseteq U\times\{\pm 1\}$ is called a sample. 
Recall that for a set~$U\subseteq \R^d$  we denote by $\hs(U) = \{H\cap U : H\in\hs_d\}$ family of all halfspaces restricted to~$U$.

\underline{\it Learning halfpaces over $U$} refers to the following search problem.
Alice's and Bob's inputs are samples $S_a,S_b \subseteq U\times\{\pm 1\}$
such that there exists a halfspace which contains all the positive examples in $S_a\cup S_b$
and does not contain any negative examples in $S_a\cup S_b$, and
their goal is to output a function $f:U\to \{\pm 1\}$ such that $f({\bm x}) = y$
for every example~$(\bm{x},y)\in S_a\cup S_b$.}
If the protocol is randomized then we require it will outputs such a function
with probability at least $2/3$.

\subsubsection*{Reductions}

All functions in this section may be partial. We denote by $\dom(f)$ the domain of the (possibly partial) function $f$.

\begin{definition}[Reduction]\label{def:reduction} We say a function $f_1:\mathcal{X}_1 \times \mathcal{Y}_1 \to \{0,1\}$ reduces to a function~$f_2:\mathcal{X}_2 \times \mathcal{Y}_2 \to \{0,1\}$ (denoted $f_1 \preceq f_2$) if there exists functions $\alpha: \mathcal{X}_1 \to \mathcal{X}_2$ and $\beta: \mathcal{Y}_1 \to \mathcal{Y}_2$ such that for all $(x,y) \in \dom(f_1)$:
\[f_1(x, y) = f_2(\alpha(x), \beta(y)).\]
We use the phrase ``reduction functions'' to refer to the functions $\alpha, \beta$. If $f_2$ is a partial function, we further require that $(\alpha(x), \beta(y)) \in \dom(f_2)$.
\end{definition}

The following results are straightforward:

\begin{observation}\label{obs:trans} For functions $f_1$, $f_2$, and $f_3$, we have $(f_1 \preceq f_2) \wedge (f_2 \preceq f_3) \implies f_1 \preceq f_3$.
\end{observation}

\begin{observation}\label{obs:cc} For functions $f_1$, $f_2$, we have $f_1 \preceq f_2  \implies R_\epsilon(f_1) \leq R_{\epsilon}(f_2)$ for all $\epsilon \geq 0$.
\end{observation}

We will also use the following basic lemma whose proof appears in in \Cref{app:misc}.
\begin{lemma}\label{lemma:and} For functions $f_1, f_2$, if $f_1 \preceq f_2$, then for any $k > 0$, we have $\prob{AND}_k \circ f_1 \preceq \prob{AND}_k \circ f_2$.
\end{lemma}

%% file: ub.tex

\subsection{Upper Bounds} \label{sec:UB}

\subsubsection{Convex Set Disjointness}

In this section, we prove the following upper bound on the communication complexity of the {\it Convex Set Disjointness} problem and its promise variant:
\begin{lemma}\label{thm:ub}
Let $U\subseteq \R^d$ with $\lvert U\rvert = n$. Then,
\begin{enumerate}
\item $\D(\prob{PromiseCSD}_U) = O(d\log d \log n)$, and
\item $\D(\prob{CSD}_U) = O(d^2\log d \log n)$.
\end{enumerate}
\end{lemma}

\Cref{thm:ub} clearly implies \Cref{thm:mainub}.
We prove \Cref{thm:ub} in two steps: 
(i)  we prove the first item by demonstrating a protocol for $\prob{PromiseCSD}_U$, and
(ii) we derive the second item by a general reduction 
that shows that any protocol for $\prob{PromiseCSD}_U$ with communication complexity $C(n,d)$
implies a protocol with communication complexity $C'(n,d) = C((2n)^{d+2},d)$ for $\prob{CSD}_U$.
Plugging $C(n,d) = O(d\log d \log n)$ then yields the second item.

\subsubsection*{An Upper Bound for $\prob{PromiseCSD}$}\label{sec:ubpromise}

We next prove the following lemma, which amounts to the first item in \Cref{thm:ub}:
\begin{lemma}\label{lem:promise}
Let $U\subseteq \R^d$ with $\lvert U\rvert = n$. Then, 
the protocol in \Cref{fig:promise} witnesses that $\D(\prob{PromiseCSD}_U) = O(d\log d \log n)$.
Furthermore, for inputs $X,Y\subseteq U$ such that $\conv(X)\cap\conv(Y)=\emptyset$,
the protocol outputs a function $h:U\to\{\pm 1\}$ such that $X\subseteq h^{-1}(-1)$
and $Y\subseteq h^{-1}(+1)$.
\end{lemma}
This function $h$ promised by the above lemma will later be used for learning halfspaces.
\begin{proof}
A complete description of the protocol is presented in Figure~\ref{fig:promise}.
The correctness  is based on the following simple observation:
\begin{observation}\label{obs:progress}
Consider the sets $U_i,X_i,Y_i$ in the ``While'' loop in item (2) of the protocol in Figure~\ref{fig:promise}.
\begin{enumerate}
\item If $\conv(X_i)\cap\conv(Y_i)=\emptyset$ 
then there is a halfspace $H\in \hs(U_i)$
such that $\lvert H \rvert \leq \lvert U_i\rvert / 2$,
and either $X_i\subseteq H$ or $Y_i\subseteq H$.
\item $X_i\cap Y_i = X_{i+1}\cap Y_{i+1}$.
\end{enumerate}
\end{observation}
The first item follows since $\conv(X_i)\cap\conv(Y_i)=\emptyset$
implies that there is a hyperplane that separates $X_i$ from  $Y_i$,
and therefore one of the two halfspaces defined by this hyperplane
contains at most half of the points in $U_i$.

The second item follows since $C \in \C_i$ either contains $X_i$ or $Y_i$. If $C \supseteq X_i$ then $X_{i+1} =X_i$ and $Y_{i+1} = Y_{i} \cap C \supseteq Y_i \cap X_i$. Otherwise, $C \supseteq Y_i$ and $X_{i+1} = X_{i} \cap C \supseteq X_i \cap Y_i$ and $Y_{i+1} = Y_i$. In both cases, $X_i\cap Y_i = X_{i+1}\cap Y_{i+1}$.

\paragraph{Correctness.}
We first assume that $\conv(X)\cap\conv(Y)=\emptyset$.
Consider iteration $i$ of the ``While'' loop.
Since $X_{i} \subseteq X$ and $Y_{i} \subseteq Y$, it holds that $\conv(X_{i})\cap\conv(Y_{i}) \subseteq \conv(X)\cap\conv(Y)=\emptyset$.
By the first item of Observation~\ref{obs:progress},
either Alice or Bob always find a container $C\in \C_i$ in item (2.2), and therefore the protocol will reach items (2.4) and (2.5). 
Since the protocol will never reach item (2.3), the ``While'' loop will eventually terminate with $|U_{i}| = 0$ and item (3) will be reached, 
outputting ``$1$'' as required.  To see that the output function $h$ satisfies $X\subseteq h^{-1}(-1), Y\subseteq h^{-1}(1)$,
note that at the $i$'th step,  $h$~is defined over all points in $U\setminus U_i$ and  satisfies $X\setminus X_i\in h^{-1}(-1)$, 
$Y\setminus Y_i \in h^{-1}(1)$.  
Thus, the requirement is met since at the last iteration $i^*$ we have $U_{i^*}=X_{i^*}=Y_{i^*}=\emptyset$.


Next, assume that $X\cap Y\neq\emptyset$.
In this case, the protocol must terminate in item (2.3) within the ``While'' loop. This is because, by the second item of Observation~\ref{obs:progress}, $|X_i\cap Y_i|$ is a positive constant for all $i$ while $\lvert U_i \rvert$ decreases, thus eventually $X_i\cap Y_i$ becomes larger than $\frac{3}{4}\lvert U_i\rvert$.
When this happens, no party can find a set $C$ satisfying the requirements of (2.2) and the protocol outputs ``$0$''.

\paragraph{Communication Complexity.}
The ``While'' loop in item (2) proceeds for at most~$O(\log n)$
iterations; this is because in each iteration $U_i$ shrinks
by a multiplicative factor of at most $3/4$.
In each of the iterations the parties exchange $\log \lvert \C_i\rvert + O(1)$ bits,
which is bounded by $O(d\log d)$ bits. 
Thus, the total number of bits communicated 
is $O(d\log d \log n)$. 

\begin{figure}
\begin{tcolorbox}
\begin{center}
{\bf An $O(d\log d \log n)$-bits deterministic protocol for $\prob{PromiseCSD}_U$}\\
\end{center}
\noindent
Let $U\subseteq \R^d$ and let $n=\lvert U\rvert$.
\\
{\em Alice's input:} $X\subseteq U$,\\
{\em Bob's input:} $Y\subseteq U$.
\\
{{\em Output:} 
if $X\cap Y\neq\emptyset$ output ``0'',\\
if $\conv(X)\cap\conv(Y)=\emptyset$ output ``1'' as well as a function $h:U\to\{\pm 1\}$
such that~$X\subseteq h^{-1}(-1)$ and $Y\subseteq h^{-1}(+1)$
($h$ will be used in our learning protocol).}
\begin{enumerate}
\item[(1)] Set $i=1$, $U_1 = U, X_1 = X, Y_1=Y$, $\eps = 1/4$, and $f$ as the empty function.
\item[(2)]  While $\lvert U_i \rvert > 0$: 
\begin{enumerate}
\item[(2.1)] Without communication, the parties agree on a set $\C_i$  of $\eps$-containers $\hs(U_i)$, 
such that  $\lvert \C_i\rvert = (d/\eps)^{O(d)}$ (as in Theorem~\ref{thm:onesidedcover}). 
\item[(2.2)] Each of Alice and Bob checks whether there is $C \in \C_i$ such that  $\lvert C\rvert \leq \frac{3}{4}\lvert U_i\rvert$
and $C$ contains their current set; namely, Alice looks for such a $C \in\C_i$ that contains $X_i$ and Bob looks for such a $C\in \C_i$ that contains $Y_i$.
\item[(2.3)] If both Alice and Bob cannot find such a $C$ then the protocol terminates with output {``0''}.
\item[(2.4)] Else, if Alice found $C$ then she communicates it to Bob (using $O(d\log d)$ bits), and the parties do:
\begin{itemize}
\item[(2.4.1)] set $X_{i+1} = X_{i}\cap C, Y_{i+1}= Y_i\cap C, U_{i+1} = U_i \cap C$,
\item[(2.4.2)] extend $h$ to $U_i\setminus C$ by setting \underline{$h(u) = 0$} for all $u\in U_i\setminus C$,
\item[(2.4.3)] increment $i\leftarrow i+1$ and go to $(2)$
\end{itemize}
\item[(2.5)] Similarly, if Bob found $C$ then he communicates it to Alice (using $O(d\log d)$ bits), and the parties do:
\begin{itemize}
\item[(2.4.1)] set $X_{i+1} = X_{i}\cap C, Y_{i+1}= Y_i\cap C, U_{i+1} = U_i \cap C$,
\item[(2.4.2)] extend $h$ to $U_i\setminus C$ by setting \underline{$h(u) = 1$} for all $u\in U_i\setminus C$,
\item[(2.4.3)] increment $i\leftarrow i+1$ and go to $(2)$
\end{itemize}
%
%
\end{enumerate}
\item [(3)] Output ``1'' and the function $h$.
\end{enumerate}
\end{tcolorbox}
\caption{A protocol for {\it Promise Convex Set Disjointness}}
\label{fig:promise}
\end{figure}

\end{proof}

\subsubsection*{From Protocols for $\prob{PromiseCSD}$ to Protocols for $\prob{CSD}$ }\label{sec:ubreduction}
The next lemma implies that a bound of $C=C(n,d)$ on the communication complexity of the promise variant 
implies a bound of $C'(n,d)=C((2n)^{d+2},d)$ on the communication complexity of the non-promise variant.
\begin{lemma}\label{lem:reduction}
For any $U\subseteq \R^d$ of size $n$ there is $V\subseteq \R^d$ of size at most $(2n)^{d+2}$
such that 
\[
\prob{CSD}_U \preceq \prob{PromiseCSD}_{V}.
\]
(Recall that ``$\preceq$'' denotes a reduction with zero communication, see Definition~\ref{def:reduction}).
\end{lemma}
Lemma~\ref{lem:reduction} implies the second item in \Cref{thm:ub}
by plugging $(2n)^{d+2}$ instead of $n$ in Lemma~\ref{lem:promise}.
Thus, Lemma~\ref{lem:promise} and Lemma~\ref{lem:reduction} imply \Cref{thm:ub}.
It therefore remains to prove Lemma~\ref{lem:reduction}.

\begin{proof}[Proof of Lemma~\ref{lem:reduction}]

The set $V$ is defined as follows: for any $S_1,S_2\subseteq U$ 
such that $\conv(S_1)\cap \conv(S_2)\neq \emptyset$
and $\lvert S_1\rvert + \lvert S_2\rvert \leq d+2$
add to $V$ (any) point $x=x(S_1,S_2)\in\conv(S_1)\cap\conv(S_2)$.
Note that indeed~$\lvert V\rvert \leq \sum_{d_1+d_2=d+2} {\lvert U\rvert \choose d_1} {\lvert U\rvert \choose d_2}\leq (2n)^{d+2}$.
Next, given inputs $X,Y\subseteq U$ for $\prob{CSD}_U$,
Alice and Bob transform them to 
\[\alpha(X) = \conv(X)\cap V \text{ and } \beta(Y)=\conv(Y)\cap V.\]

\paragraph{Validity.}
To establish the validity of this reduction we need to show that
\begin{align*}
\conv(X)\cap\conv(Y) =\emptyset &\implies \conv(\alpha(X))\cap\conv(\beta(Y)) = \emptyset, \text{ and }\\
\conv(X)\cap \conv(Y) \neq\emptyset &\implies  \alpha(X)\cap\beta(Y)\neq\emptyset.
\end{align*}
Indeed, if $\conv(X)\cap\conv(Y) =\emptyset$
then also $\conv(\alpha(X))\cap\conv(\alpha(Y)) = \emptyset$ 
(because $\alpha(X)\subseteq \conv(X)$ and $\beta(Y)\subseteq \conv(Y)$).

The second assertion follows from \Cref{lem:bicaratheodory} which we next recall:
\begin{proposition*}[\Cref{lem:bicaratheodory} restatement]
Let $X, Y\subseteq \R^d$ such that $\conv(X)\cap \conv(Y)\neq\emptyset$.
Then $\conv(S_1)\cap \conv(S_2)\neq\emptyset$ for some $S_1\subseteq X, S_2\subseteq Y$
such that $\lvert S_1\rvert + \lvert S_2\rvert \leq d+2$.
\end{proposition*}

To see how this implies the second assertion, assume that $\conv(X)\cap \conv(Y) \neq\emptyset$. By \Cref{lem:bicaratheodory}, there exists $S_1\subseteq X, S_2\subseteq Y$ with $\lvert S_1\rvert + \lvert S_2\rvert \leq d+2$ and $\conv(S_1)\cap \conv(S_2)\neq\emptyset$. By construction, $V$ contains a point $x=x(S_1,S_2)$ in $\conv(S_1)\cap\conv(S_2)$. It holds that $x \in \conv(S_1) \cap V \subseteq \conv(X) \cap V= \alpha(X)$ and $x \in \conv(S_2) \cap V \subseteq \conv(Y) \cap V= \beta(Y)$. Hence, $\alpha(X)\cap\beta(Y)\neq\emptyset$, as claimed.
\end{proof}

\subsubsection{Learning Halfspaces}\label{sec:hsub}

We next prove the following upper bound for learning halfspaces.

\begin{theorem*}[\Cref{thm:ublearning} restatement]
Let $d,n\in\mathbb{N}$, and let $U\subseteq \R^d$ be a domain with $n$ points. 
Then, there exists a deterministic protocol for learning $\hs(U)$ with communication complexity~$O(d\log d\log n)$.
\end{theorem*}

\begin{proof}
We present a learning protocol which relies on \Cref{thm:ub} and uses the protocol in \Cref{fig:promise} as a black-box.
The learning protocol is presented in \Cref{fig:learning}.

\begin{figure}
\begin{tcolorbox}
\begin{center}
{\bf An $O(d\log d \log n)$-bits deterministic learning protocol for halfspaces}\\
\end{center}
\noindent
Let $U\subseteq \R^d$ and let $n=\lvert U\rvert$.
\\
{\em Alice's input:} a sample $S_a \subseteq U\times\{\pm 1\}$,\\
{\em Bob's input:} a sample $S_b \subseteq U\times\{\pm 1\}$.\\
(It is assumed that there exists a separating hyperplane between the positive and negative examples
in $S_a\cup S_b$).\\
{\em Output:} a function $h:U\to\{\pm 1\}$ such that $h(\bm{x})=y$ for every $(\bm{x},y)\in S_a\cup S_b$.
\begin{enumerate}
\item [(1)] Apply the protocol from \Cref{fig:promise} on inputs $X^-,Y^+$,
where $X^- = \{\bm{u}: (\bm{u},-1)\in S_a\}$ and $Y^+ = \{\bm{u}: (\bm{u},+1)\in S_b\}$.
\begin{itemize}
\item[(1.1)] If the protocol outputted "0" then output ``Error''.
\item[(1.2)] Else,  let $g: U\to \{\pm1\}$ denote the function outputted 
by the protocol, such that $g(\bm{u})=+1 $ for every $\bm{u}\in Y^+$ and $g(\bm{u})=-1$ for every $\bm{u}\in X^-$.
\end{itemize}
\item [(2)] Apply the protocol from \Cref{fig:promise} on inputs $X^+,Y^-$,
where $X^+ = \{\bm{u}: (\bm{u},+1)\in S_a\}$ and $Y^+ = \{\bm{u}: (\bm{u},-1)\in S_b\}$.
\begin{itemize}
\item[(2.1)] If the protocol outputted "0" then output ``Error''.
\item[(2.2)] Else,  let $f: U\to \{\pm1\}$ denote the function outputted 
by the protocol, such that $f(\bm{u})=+1 $ for every $\bm{u}\in X^+$ and $f(\bm{u})=-1$ for every $\bm{u}\in Y^-$.
(note that $f$ is actually the negation of the output function.)
\end{itemize}
\item[(3)] Let $F^+=f^{-1}(+1), F^- = f^{-1}(-1)$ and $G^+=g^{-1}(+1), G^- = g^{-1}(-1)$.
(Note that these 4 sets are known to both Alice and Bob.)
\item[(4)] Alice transmits to Bob using $O(d\log n)$ bits an indicator $I_{+-}:U\to \{\pm 1\}$ 
of a halfspace in $\hs(U)$ which separates \underline{her} positive and negative examples in~$F^+\cap G^-$;
namely, $I_{+-}(\bm{u}) = b$ for every $\bm{u}\in F^+\cap G^-$ such that $(\bm{u},b)\in S_a$.
\item[(5)] Bob transmits to Alice using $O(d\log n)$ bits an indicator $I_{-+}:U\to \{\pm 1\}$ 
of a halfspace in $\hs(U)$ which separates \underline{his} positive and negative examples in $F^-\cap G^+$;
namely, $I_{-+}(\bm{u}) = b$ for every $\bm{u}\in F^-\cap G^+$ such that $(\bm{u},b)\in S_b$.
\item[(6)] Alice and Bob output the function $h$ defined by
\[ h(u) =
\begin{cases}
+1 			&\bm{u}\in F^+\cap G^+,\\
-1			&\bm{u}\in F^-\cap G^-,\\
I_{+-}(u)	&\bm{u}\in F^+\cap G^-,\\
I_{-+}(u)	&\bm{u}\in F^-\cap G^+.
\end{cases}
\]
\end{enumerate}
\end{tcolorbox}
\caption{A protocol for {\it Promise Convex Set Disjointness}}
\label{fig:learning}
\end{figure}

\paragraph{Analysis.}
First, note that the communication complexity is at most $O(d\log d \log n)$ bits:
indeed, there is no communication in steps (3) and (6),
each of steps (1) and (2) involves an application of the protocol from \Cref{fig:promise} which costs $O(d\log d\log n)$ bits, 
and each of steps (4) and (5) involves transmitting a separator from~$\hs(U)$ which costs $O(d \log n)$ bits (since~$\lvert \hs(U)\rvert \leq O(n^d)$,
see {\em e.g.}~\citep{Gartner94vapnik}).

As for correctness, 
note that since it is assumed that the negative and positive examples in $S_a\cup S_b$
are separated by a hyperplane, \Cref{thm:ub} implies that the functions $f,g$ which are outputted in steps (1) and (2) satisfy:
\begin{itemize}
\item $f(\bm{u})=+1$ for every $(\bm{u},+1)\in S_a$ and $f(\bm{u})=-1$ for every $(\bm{u},-1)\in S_b$, and similarly 
\item $g(\bm{u})=-1$ for every $(\bm{u},-1)\in S_a$ and $g(\bm{u})=+1$ for every $(\bm{u},+1)\in S_b$. 
\end{itemize}
We will show that the $h$ (the function outputted by the protocol) classifies correctly
each of the regions $F^+\cap G^+, F^-\cap G^-, F^+\cap G^-,$ and $F^-\cap G^+$ (the definition of these regions appears in the protocol).
Since these 4 regions cover $U$, it will follow that $h$ classifies correctly all examples.
Indeed $F^+\cap G^+$ contains only positive examples and $F^-\cap G^-$ contains only negative examples,
therefore $h$ classifies correctly these regions.
As for $F^+\cap G^-$ and $F^-\cap G^+$, note that $F^+\cap G^-$ contains only examples
in $S_a$ and $F^-\cap G^+$ contains only examples in~$S_b$.
Thus, $I_{+-}$ classifies correctly every example in $F^+\cap G^-$
and $I_{-+}$ classifies correctly every example in $F^-\cap G^+$.
It therefore follows that $h$ classifies correctly also these regions.
\end{proof}

\paragraph{Remark.}
Note that the above protocol actually learns a more general problem than halfspaces:
indeed, let $S_a^+, S_a^-$  denote Alice's positive and negative examples respectively,
and let $S_b^+, S_b^-$  denote Bob's positive and negative examples respectively.
The protocol will output a consistent function $h$ for as long as each of the pairs $S_a^+$ and $S_a^-$, $S_b^+$ and $S_b^-$,
$S_a^+$ and $S_b^-$, and $S_b^+$ and $S_a^-$ can be separated by a hyperplane (possibly a different hyperplane for every pair).
However it is not necessary that there will be a single hyperplane separating all positive examples
from all negative examples.

%% file: lb.tex

\subsection{Lower Bounds} \label{sec:LB}

\subsubsection{Convex Set Disjointness}
In this section we prove a lower bound on the randomized communication complexity of~$\prob{PromiseCSD}$. 
This implies the same lower bound for $\prob{CSD}$, and therefore yields \Cref{thm:mainlb}.
More precisely, we prove that

\begin{theorem}\label{thm:LB} Let $n, d > 0$ be integers. There is a set $U \subseteq \mathbb{R}^d$ such that $\lvert{U}\rvert = n$ and~$R(\prob{PromiseCSD}_{U}) \geq \Omega(d \log (n/d))$. 
\end{theorem} 

The key ingredient in the proof of \Cref{thm:LB} is the following reduction:

\begin{lemma}\label{lemma:redn} For any integers $c, k > 0$, there is a set $U \subseteq \mathbb{R}^{3c}$ such that $\lvert{U}\rvert = 2^{k}c$ and 
\[\prob{DISJ}_{ck} \preceq \prob{PromiseCSD}_{U}.\]
\end{lemma}
 
 We prove \Cref{lemma:redn} below. Assuming \Cref{lemma:redn}, the following argument proves \Cref{thm:LB}.
 
\begin{proof}[Proof of \Cref{thm:LB}]
Fix $d$ and $n$. 
Set $c = d/3$ and set $k$ such that $n = 2^k c$. We assume without loss of generality that $k, c$ are positive integers.
By \Cref{lemma:redn}, there is set $U \subseteq \mathbb{R}^d$, $\lvert{U}\rvert = n$ such that 
\[\prob{DISJ}_{ck} \preceq \prob{PromiseCSD}_{U}.\]

Using the the well known fact that $R(\prob{DISJ}_{m}) \geq \Omega(m)$ (see, {\em e.g.}, \cite{kalyanasundaram1992probabilistic}), and \Cref{obs:cc}, it follows that 
\[R(\prob{PromiseCSD}_{U}) \geq \Omega(ck) = \Omega(d \log (n/d)).\]


\end{proof}

\subsubsection*{Proof of \Cref{lemma:redn}} \label{sec:lbredn}

Let $c, k > 0$ be arbitrary. 
To prove \Cref{lemma:redn}, we show that there exist sets $U\subseteq \mathbb{R}^{3c},V \subseteq \mathbb{R}^{2}$ such that $\lvert{U}\rvert = 2^k c$ and $\lvert{V}\rvert = 2^k$ such that the following sequence of reductions holds 

\[\prob{DISJ}_{ck} \preceq \prob{AND}_{c} \circ \prob{DISJ}_k \preceq \prob{AND}_{c} \circ \prob{PromiseCSD}_{V} \preceq \prob{PromiseCSD}_{U}.\]

Each of these reductions is proved separately below. \Cref{lemma:redn} then follows using \Cref{obs:trans}.

\paragraph{Proving $\prob{DISJ}_{ck} \preceq \prob{AND}_{c} \circ \prob{DISJ}_k $.} 
The first reduction in our sequence is essentially using the fact that $\prob{DISJ}_{m}$ can be viewed as an $\prob{AND}$ of $m$ simpler functions.

\begin{lemma}\label{lemma:dirsumdisj} $\prob{DISJ}_{ck} \preceq \prob{AND}_c \circ \prob{DISJ}_k$
\end{lemma}
\begin{proof}  Let $\bm{x^*}, \bm{y^*} \in \{0,1\}^{ck}$ be an input for $\prob{DISJ}_{ck}$.  We can view $\bm{x^*}$ as a vector $\bm{x}^{(c)}$ with entries in $\mathbb{R}^k$. Precisely, $\bm{x_i}$ (respectively $\bm{y_i}$) is the $((i-1)k + 1)^{\text{st}}$ to $(ik)^{\text{th}}$ coordinates of $\bm{x^*}$ (resp. $\bm{y^*}$).  Let the reduction function $\alpha$ (resp. $\beta$) be the function that takes $\bm{x^*}$ to $\bm{x}$ (resp. $\bm{y^*}$ to $\bm{y}$). Note that:
\begin{align*}
\prob{DISJ}_{ck}(\bm{x^*}, \bm{y^*}) = 0 &\iff \exists i \in [ck] \colon  x^*_i = y^*_i = 1  \\
&\iff \exists i \in [c], j\in [k] \colon  x_{ij} = y_{ij} = 1\\
&\iff \exists i \in [c] \colon \prob{DISJ}_{k}(\bm{x_i}, \bm{y_i}) = 0\\
&\iff \left( \bigwedge_{i =1}^c \prob{DISJ}_{k}(\bm{x_i}, \bm{y_i}) \right) = 0\\
&\iff \prob{AND}_c \circ \prob{DISJ}_k(\bm{x}, \bm{y}) = 0 \\
&\iff \prob{AND}_c \circ \prob{DISJ}_k(\alpha(\bm{x^*}), \beta(\bm{y^*})) = 0.
\end{align*}
\end{proof}

%

\paragraph{Proving $\prob{AND}_{c} \circ \prob{DISJ}_k \preceq \prob{AND}_{c} \circ \prob{PromiseCSD}_{V}$.}

By \Cref{lemma:and}, the following result is sufficient:

\begin{lemma}\label{lemma:pcsdcomplete} For all $k > 0$, there exists $V \subseteq \mathbb{R}^2$, $\lvert{V}\rvert = 2^k$ such that $\prob{DISJ}_{k} \preceq \prob{PromiseCSD}_{V}$.
\end{lemma}

\begin{proof} We define the set $V$ to consist of $2^k$ points on the unit circle in $\mathbb{R}^2$. 
The crucial property satisfied by these set of points is that every $\bm{v}\in V$ can be separated by a line from~$V\setminus\{\bm{v}\}$
({\em i.e.},\ these points are in {\it convex position}).
Let us index the points in $V$ by the vectors in $\{0,1\}^k$, {\em i.e.},  $V = \{\bm{v}_{\bm{x}} \mid \bm{x} \in \{0,1\}^k\}$ (see \autoref{fig:4}).

We next define the functions $\alpha,\beta$ which witness the desired reduction.
Define $\alpha : \{0,1\}^k \to \mathbbm{2}^V$  by
\[\alpha(\bm{x}) = \{\bm{v}_{\bm{x}}\}.\]
Next, define $\beta : \{0,1\}^k \to \mathbbm{2}^V$ as
\[\beta(\bm{y}) = \{\bm{v}_{\bm{z}}  \text{ for $\bm{z} \in \{0,1\}^k$ such that } \exists i \in [k] : z_i = y_i = 1\}.\]

Observe that for every input $\bm{x} \in \{0,1\}^k$, the set $\alpha(\bm{x}) = \{ \bm{v}_{\bm{x}} \}$ is a singleton. Thus, for every possible $\bm{y} \in \{0,1\}^k$, it is either the case that $\bm{v}_{\bm{x}} \in \beta(\bm{y})$, or else, since $\bm{x} \in V$ and $\beta(\bm{y}) \subseteq V$, and due to the crucial property described above, it is the case that~$\bm{v}_{\bm{x}} \notin \conv(\beta(\bm{y}))$. Equivalently, it is either the case that $\alpha(\bm{x}) \cap  \beta(\bm{y})  \neq \emptyset$ or that $\conv(\alpha(\bm{x})) \cap  \conv(\beta(\bm{y}))  = \emptyset$, thus the sets $\alpha(\bm{x})$ and $\beta(\bm{y})$ are in the domain of $\prob{PromiseCSD}_V$.

We have 
\begin{align*}
\prob{DISJ}_{k}(\bm{x}, \bm{y}) = 0 &\iff \exists i \in [k] \colon  x_i = y_i = 1  \\
&\iff \alpha(\bm{x}) \cap  \beta(\bm{y})  \neq \emptyset \\
&\iff \prob{PromiseCSD}_V(\alpha(\bm{x}), \beta(\bm{y})) = 0.
\end{align*}
\end{proof}

\paragraph{Proving $\prob{AND}_{c} \circ \prob{PromiseCSD}_{V} \preceq \prob{PromiseCSD}_{U}$.}

\begin{lemma}\label{lemma:dirsumcsd} Let $V \subseteq \mathbb{R}^2$, $\lvert{V}\rvert = m$. For all integers $c>0$, there is a set $U \subseteq \mathbb{R}^{3c}$ of size~$c\cdot m$ such that
\[\prob{AND}_c \circ \prob{PromiseCSD}_V \preceq \prob{PromiseCSD}_{U}.\]
\end{lemma}

\begin{proof} 

We embed each of the $c$ copies of $\prob{PromiseCSD}_V$ in a disjoint triplet of coordinates of~$\R^{3c}$.
Formally, for $j \in [c]$, define the $j^{\text{th}}$ `lift' function $g_j :\mathbb{R}^2 \to \mathbb{R}^{3c}$ as: 
\[g_{j}((x_1,x_2)) = (\underbrace{0,0,\cdots,0}_{3(j-1) \text{ times}}, x_1, x_2, 1, \underbrace{0,0,\cdots,0}_{3(c-j) \text{ times}}).\]
Define the set $ U = \{g_j(v) \mid j \in [c], v \in V\}$.

Let $\bm{X}^{(c)} , \bm{Y}^{(c)}$ be an input for $\prob{AND}_c \circ \prob{PromiseCSD}_V$. Define: 
\[\alpha(\bm{X}) =  \bigcup_{j =1}^c g_j(X_j)\quad \quad \quad \quad \beta(\bm{Y}) =  \bigcup_{j =1}^c g_j(Y_j).\]
{(Recall that $X_j,Y_j$ denote the $j$'th copies of $\bm{X}^{(c)},\bm{Y}^{(c)}$ respectively.)}
We prove that $\alpha$, $\beta$ define the desired reduction. First, assume that $\prob{AND}_c \circ \prob{PromiseCSD}_V(\bm{X}, \bm{Y}) = 1$, that is, $\forall j \in [c], \conv(X_j) \cap \conv(Y_j) = \emptyset$.  By the hyperplane separation theorem, for every $j \in [c]$ there exists an affine function $l_j: \R^2 \to \R$ of the form $l_j((x_1, x_2)) =  l_{j} + l'_{j}x_1 + l''_{j} x_2$ such that~$l_j(x) > 0$ for all $x \in X_j$, while $l_j(y) < 0$ for all $y \in Y_j$. 

Define the affine function $l : \R^{3c} \to \R$ by $l((x_1, x_2, \cdots, x_{3c})) = \sum_{i\in [c]} l_{j} x_{3j} +  l'_{j} x_{3j-2} + l''_{j} x_{3j-1}$. Observe that for all $j \in [c]$, we have $\forall (x_1, x_2) \in \mathbb{R}^2 \colon l(g_j((x_1, x_2))) = l_j((x_1, x_2))$.  This implies that $l(x) > 0$ for all $x \in \alpha(\bm{X})$, while $l(y) < 0$ for all $y \in \beta(\bm{Y})$. Thus,~$\alpha(\bm{X}) \cap  \beta(\bm{Y}) = \emptyset$, implying $\prob{PromiseCSD}_{U}(\alpha(\bm{X}),\beta(\bm{Y})) = 1$.

For the other direction, assume that $\prob{AND}_c \circ \prob{PromiseCSD}_V(\bm{X}, \bm{Y}) = 0$, that is, $\exists j \in [c], z \in V \colon\; z \in X_j \cap Y_j$. Then, $g_j(z) \in \alpha(\bm{X}) \cap  \beta(\bm{Y})$, implying $\alpha(\bm{X}) \cap  \beta(\bm{Y}) \neq \emptyset$ and therefore also $\prob{PromiseCSD}_{U}(\alpha(\bm{X}),\beta(\bm{Y})) = 0$.

\end{proof}
\subsubsection{Learning Halfspaces}\label{sec:hslb}

\begin{theorem*}[\Cref{thm:lblearning} restatement] 
Let $d,n \in\mathbb{N}$. 
Then, there exists a domain $U\subseteq \R^d$ with $n$ points such that every 
(possibly improper and randomized) protocol that learns $\hs(U)$ must transmit at least $\Omega(d\log(n/d))$ bits 
of communication.
\end{theorem*}

\begin{proof}
This is a corollary of \Cref{thm:LB}: let $U\subseteq \R^d$ be as in the conclusion of \Cref{thm:LB}.
We claim that every protocol that learns $\hs(U)$ can be used to decide~$\prob{PromiseCSD}_{U}$.
Indeed, let $X,Y$ be inputs to $\prob{PromiseCSD}_{U}$. Alice and Bob apply the learning protocol on the samples $X\times\{+1\}$
and $Y\times\{-1\}$. (i) If $\conv(X)\cap\conv(Y)=\emptyset$ then $X,Y$ can be separated by a hyperplane
and the protocol will output a function $h:U\to\{\pm 1\}$ such that $h(\bm{u})=+1$ for every $\bm{u}\in X$ and $h(\bm{u})=-1$
for every $\bm{u}\in Y$. (ii) In the other case, if $X\cap Y = \emptyset$ then there exists no such function and therefore
the learning protocol must output ``Error''.
Therefore, by \Cref{thm:LB}, every such learning protocol must transmit at least $\Omega(d\log(n/d))$ bits.

\end{proof}

\section{Summary and Future Research}

We established bounds on the communication complexity of convex set disjointness (equivalently, LP feasibility)
	and learning halfspaces over a  domain of $n$ points in $\R^d$.

For learning halfspaces we establish a bound of $\tilde\Theta(d\log n)$, which is tight up to a $\log d$ factor.
	Our upper bound is achieved by an improper protocol (i.e.\ it returns a classifier which is not necessarily a halfspace).
	It would be interesting to determine whether a similar bound can be achieved by a proper learning protocol.

For Convex Set Disjointness, the gap between our lower and upper bounds is more significant:
	$\tilde O(d^2\log n)$ versus $\Omega(d\log n)$, and it would be interesting to tighten it.

Another interesting direction is to further explore the halfspace container lemma
	which we used (e.g.\ improve the bound, find other natural VC classes which satisfy a similar statement, etcetera.)

\section*{Acknowledgements}
We thank Noga Alon, Sepehr Assadi, and Shachar Lovett for insightful discussions and comments.

%% file: app1.tex
\section{Missing proofs} \label{app:misc}

\begin{proof}[Proof of \Cref{lemma:and}] Since $f_1 \preceq f_2$, we know that there exists reduction functions $\alpha, \beta$ such that for all $(x,y) \in \dom(f_1)$:
\[f_1(x, y) = f_2(\alpha(x), \beta(y)).\]
Define:
\begin{align*}
\alpha^*(\bm{x}^{(k)}) &= (\alpha(x_1), \alpha(x_2), \cdots, \alpha(x_k)),\\
\beta^*(\bm{y}^{(k)}) &= (\beta(y_1), \beta(y_2), \cdots, \beta(y_k)).
\end{align*}
Note that
\begin{align*}
\prob{AND}_k \circ f_1(\bm{x},\bm{y}) &= \bigwedge_{i \in [k]}  f_1(x_i,y_i)  \\
&= \bigwedge_{i \in [k]}  f_2(\alpha(x_i), \beta(y_i)) = \prob{AND}_k \circ f_2(\alpha^*(\bm{x}),\beta^*(\bm{y})).
\end{align*}

\end{proof}

%% file: main.bbl
\begin{thebibliography}{24}
\providecommand{\natexlab}[1]{#1}
\providecommand{\url}[1]{\texttt{#1}}
\expandafter\ifx\csname urlstyle\endcsname\relax
  \providecommand{\doi}[1]{doi: #1}\else
  \providecommand{\doi}{doi: \begingroup \urlstyle{rm}\Url}\fi

\bibitem[Balcan et~al.(2012)Balcan, Blum, Fine, and Mansour]{Balcan12dist}
Maria{-}Florina Balcan, Avrim Blum, Shai Fine, and Yishay Mansour.
\newblock Distributed learning, communication complexity and privacy.
\newblock In \emph{{COLT} 2012 - The 25th Annual Conference on Learning Theory,
  June 25-27, 2012, Edinburgh, Scotland}, pages 26.1--26.22, 2012.
\newblock URL
  \url{http://www.jmlr.org/proceedings/papers/v23/balcan12a/balcan12a.pdf}.

\bibitem[Balogh et~al.(2018)Balogh, Morris, and Samotij]{Balogh18containers}
Jozsef Balogh, Robert Morris, and Wojciech Samotij.
\newblock The method of hypergraph containers, 2018.

\bibitem[{Blumer} et~al.(1989){Blumer}, {Ehrenfeucht}, {Haussler}, and
  {Warmuth}]{BEHW89}
A.~{Blumer}, A.~{Ehrenfeucht}, D.~{Haussler}, and M.~K. {Warmuth}.
\newblock {Learnability and the Vapnik-Chervonenkis dimension.}
\newblock \emph{{J. Assoc. Comput. Mach.}}, 36\penalty0 (4):\penalty0 929--965,
  1989.
\newblock ISSN 0004-5411.
\newblock \doi{10.1145/76359.76371}.

\bibitem[Carath\'eodory(1907)]{Caratheodory07}
C.~Carath\'eodory.
\newblock {\"Uber den {V}ariabilit\"atsbereich der {K}oeffizienten von
  {P}otenzreihen, die gegebene {W}erte nicht annehmen}.
\newblock \emph{Math. Ann.}, 64\penalty0 (1):\penalty0 95--115, 1907.
\newblock ISSN 0025-5831.
\newblock \doi{10.1007/BF01449883}.
\newblock URL \url{https://doi.org/10.1007/BF01449883}.

\bibitem[Chen et~al.(2016)Chen, Balcan, and Chau]{Chen16boosting}
Shang{-}Tse Chen, Maria{-}Florina Balcan, and Duen~Horng Chau.
\newblock Communication efficient distributed agnostic boosting.
\newblock In \emph{Proceedings of the 19th International Conference on
  Artificial Intelligence and Statistics, {AISTATS} 2016, Cadiz, Spain, May
  9-11, 2016}, pages 1299--1307, 2016.
\newblock URL \url{http://jmlr.org/proceedings/papers/v51/chen16e.html}.

\bibitem[Clarkson(1988)]{Clarkson88queries}
Kenneth~L. Clarkson.
\newblock A randomized algorithm for closest-point queries.
\newblock \emph{SIAM J. Comput.}, 17\penalty0 (4):\penalty0 830--847, August
  1988.
\newblock ISSN 0097-5397.
\newblock \doi{10.1137/0217052}.
\newblock URL \url{http://dx.doi.org/10.1137/0217052}.

\bibitem[Clarkson(1995)]{Clarkson95lasvegas}
Kenneth~L. Clarkson.
\newblock Las vegas algorithms for linear and integer programming when the
  dimension is small.
\newblock \emph{J. ACM}, 42\penalty0 (2):\penalty0 488--499, March 1995.
\newblock ISSN 0004-5411.
\newblock \doi{10.1145/201019.201036}.
\newblock URL \url{http://doi.acm.org/10.1145/201019.201036}.

\bibitem[Dagan et~al.(2019)Dagan, Kur, and Shamir]{Dagan19space}
Yuval Dagan, Gil Kur, and Ohad Shamir.
\newblock Space lower bounds for linear prediction.
\newblock In \emph{COLT, to appear}, volume abs/1902.03498, 2019.

\bibitem[Daum{\'{e}~III} et~al.(2012)Daum{\'{e}~III}, Phillips, Saha, and
  Venkatasubramanian]{Daume12efficient}
Hal Daum{\'{e}~III}, Jeff~M. Phillips, Avishek Saha, and Suresh
  Venkatasubramanian.
\newblock Efficient protocols for distributed classification and optimization.
\newblock In \emph{Algorithmic Learning Theory - 23rd International Conference,
  {ALT} 2012, Lyon, France, October 29-31, 2012. Proceedings}, pages 154--168,
  2012.
\newblock \doi{10.1007/978-3-642-34106-9_15}.
\newblock URL \url{https://doi.org/10.1007/978-3-642-34106-9_15}.

\bibitem[Feige et~al.(1994)Feige, Peleg, Raghavan, and Upfal]{Feige94computing}
Uriel Feige, David Peleg, Prabhakar Raghavan, and Eli Upfal.
\newblock Computing with noisy information.
\newblock \emph{SIAM Journal on Computing}, 23\penalty0 (5):\penalty0
  1001--1018, 1994.

\bibitem[G{\"{a}}rtner and Welzl(1994)]{Gartner94vapnik}
Bernd G{\"{a}}rtner and Emo Welzl.
\newblock Vapnik-chervonenkis dimension and (pseudo-)hyperplane arrangements.
\newblock \emph{Discrete {\&} Computational Geometry}, 12:\penalty0 399--432,
  1994.
\newblock \doi{10.1007/BF02574389}.
\newblock URL \url{https://doi.org/10.1007/BF02574389}.

\bibitem[Goodman and O'Rourke(2004)]{Goodman04handbook}
Jacob~E. Goodman and Joseph O'Rourke, editors.
\newblock \emph{Handbook of Discrete and Computational Geometry, Second
  Edition}.
\newblock Chapman and Hall/CRC, 2004.
\newblock ISBN 978-1-58488-301-2.
\newblock \doi{10.1201/9781420035315}.
\newblock URL \url{https://doi.org/10.1201/9781420035315}.

\bibitem[Haussler(1995)]{haussler1995sphere}
David Haussler.
\newblock Sphere packing numbers for subsets of the boolean n-cube with bounded
  vapnik-chervonenkis dimension.
\newblock \emph{Journal of Combinatorial Theory, Series A}, 69\penalty0
  (2):\penalty0 217--232, 1995.

\bibitem[Haussler and Welzl(1986)]{haussler1986epsilon}
David Haussler and Emo Welzl.
\newblock Epsilon-nets and simplex range queries.
\newblock In \emph{CG}, pages 61--71, 1986.

\bibitem[Kalyanasundaram and Schintger(1992)]{kalyanasundaram1992probabilistic}
Bala Kalyanasundaram and Georg Schintger.
\newblock The probabilistic communication complexity of set intersection.
\newblock \emph{SIAM Journal on Discrete Mathematics}, 5\penalty0 (4):\penalty0
  545--557, 1992.

\bibitem[Kane et~al.(2019)Kane, Livni, Moran, and
  Yehudayoff]{kane17communication}
Daniel~M. Kane, Roi Livni, Shay Moran, and Amir Yehudayoff.
\newblock On communication complexity of classification problems.
\newblock In \emph{COLT, to appear}, volume abs/1711.05893, 2019.

\bibitem[Kone{\v c}n{\'y} et~al.(2016)Kone{\v c}n{\'y}, McMahan, Yu, Richtarik,
  Suresh, and Bacon]{Konecny16federated}
Jakub Kone{\v c}n{\'y}, H.~Brendan McMahan, Felix~X. Yu, Peter Richtarik,
  Ananda~Theertha Suresh, and Dave Bacon.
\newblock Federated learning: Strategies for improving communication
  efficiency.
\newblock In \emph{NIPS Workshop on Private Multi-Party Machine Learning},
  2016.
\newblock URL \url{https://arxiv.org/abs/1610.05492}.

\bibitem[Kushilevitz and Nisan(1997)]{Kushilevitz97book}
Eyal Kushilevitz and Noam Nisan.
\newblock \emph{Communication complexity}.
\newblock Cambridge University Press, 1997.
\newblock ISBN 978-0-521-56067-2.

\bibitem[Lov{\u a}sz and Saks(1993)]{lovasz93communication}
L{\'a}szl{\'o} Lov{\u a}sz and Michael Saks.
\newblock Communication complexity and combinatorial lattice theory.
\newblock \emph{Journal of Computer and System Sciences}, 47\penalty0
  (2):\penalty0 322 -- 349, 1993.
\newblock ISSN 0022-0000.
\newblock \doi{https://doi.org/10.1016/0022-0000(93)90035-U}.
\newblock URL
  \url{http://www.sciencedirect.com/science/article/pii/002200009390035U}.

\bibitem[Rosenblatt(1958)]{Rosenblatt58perceptron}
F.~Rosenblatt.
\newblock The perceptron: A probabilistic model for information storage and
  organization in the brain.
\newblock \emph{Psychological Review}, pages 65--386, 1958.

\bibitem[Vapnik and Chervonenkis(2015)]{vapnik2015uniform}
Vladimir~N Vapnik and A~Ya Chervonenkis.
\newblock On the uniform convergence of relative frequencies of events to their
  probabilities.
\newblock In \emph{Measures of Complexity}, pages 11--30. Springer, 2015.

\bibitem[Vempala et~al.(2019)Vempala, Wang, and
  Woodruff]{vempala19optimization}
Santosh~S. Vempala, Ruosong Wang, and David~P. Woodruff.
\newblock The communication complexity of optimization, 2019.

\bibitem[Viola(2013)]{Viola13addition}
Emanuele Viola.
\newblock The communication complexity of addition.
\newblock In \emph{SODA}, pages 632--651, 2013.

\bibitem[Yao(1979)]{Yao79}
Andrew Chi-Chih Yao.
\newblock Some complexity questions related to distributive computing
  (preliminary report).
\newblock In \emph{STOC}, pages 209--213, 1979.

\end{thebibliography}
